\documentclass{lmcs}
\pdfoutput=1

\usepackage{lastpage}
\lmcsdoi{21}{1}{3}
\lmcsheading{}{\pageref{LastPage}}{}{}%
{Jul.~19,~2023}{Jan.~10,~2025}{}

\usepackage[utf8]{inputenc}

\usepackage{hyperref}

\usepackage{amsmath}
\usepackage{amssymb}

\usepackage{xspace}

\usepackage{tikz}
\usetikzlibrary{decorations.pathreplacing,positioning}
\usetikzlibrary{arrows,automata}
\tikzset{>=stealth, shorten >=1pt}
\tikzset{every edge/.style = {thick, ->, draw}}
\tikzset{every loop/.style = {thick, ->, draw}}
\pgfdeclarelayer{background}
\pgfsetlayers{background,main}

\newcommand{\expansion}{e}

\newcommand{\contcard}{\cont}
\newcommand{\mcgame}{\mathcal{G}}

\newcommand{\last}{\mathrm{rcnt}}

\newcommand{\set}[1]{\{#1\}}
\newcommand{\nats}{\mathbb{N}}

\renewcommand{\epsilon}{\varepsilon}

\renewcommand{\phi}{\varphi}
\newcommand{\size}[1]{|#1|}

\newcommand\Nat{\nats}

\newcommand{\pow}[1]{2^{#1}}
\newcommand{\cont}{\mathfrak{c}}

\newcommand{\cceq}{\mathop{::=}}

\newcommand{\tsys}{\mathcal{T}}

\newcommand{\initmark}{I}

\newcommand{\F}{\mathop{\mathbf{F}\vphantom{a}}\nolimits}
\newcommand{\G}{\mathop{\mathbf{G}\vphantom{a}}\nolimits}
\DeclareMathOperator{\U}{\mathbf{U}}
\newcommand{\X}{\mathop{\mathbf{X}\vphantom{a}}\nolimits}

\newcommand{\ltl}{{LTL}\xspace}
\newcommand{\ctlstar}{{CTL$^*$}\xspace}
\newcommand{\hyltl}{{HyperLTL}\xspace}
\newcommand{\hyqptl}{{HyperQPTL}\xspace}
\newcommand{\hyctlstar}{{HyperCTL$^*$}\xspace}
\newcommand{\fol}{{FO}$[<]$\xspace}

\newcommand{\teamltl}{{TeamLTL}\xspace}

\newcommand\FO{\textup{FO}}

\newcommand\simpleHyperLTL{\encPhi {\FOw}}

\newcommand{\suffix}[2]{#1[#2,\infty)}
\newcommand{\var}{\mathcal{V}}
\newcommand{\ap}[0]{\mathrm{AP}}
\newcommand\AP{\textup{AP}}
\newcommand\Free{\textup{Free}}
\newcommand\depth{\textit{depth}}

\newcommand{\myquot}[1]{``#1''}

\newcommand\FOw{\FO[\le]}
\newcommand\varfo{\textit{Var}}
\newcommand\pos{\textit{pos}}
\newcommand\Sfo[1]{\Sigma_{#1}(\FOw)}

\newcommand\enc[2]{\mathit{enc}(#1,#2)}
\newcommand\encn[2]{\mathit{enc}_{#2}(#1)}
\newcommand\tn[2]{#1^{(#2)}}
\newcommand\nun[2]{#1^{(#2)}}
\newcommand\encPhi[1]{\mathit{enc}(#1)}

\newcommand\simpl[1]{\widehat {#1}}

\newcommand\Sign[1]{\Sigma_{#1}}
\newcommand\Pin[1]{\Pi_{#1}}
\newcommand\Tl{T_\ell}
\newcommand\Tr {T_r}

\newcommand\phib{\varphi_{\mathit{bd}}}
\newcommand\Vl{V_\ell}
\newcommand\Vr{V_r}
\newcommand\Sigmal{\Sigma_\ell}
\newcommand\Sigmar{\Sigma_r}

\newcommand\argl{\textup{\texttt{arg1}}}
\newcommand\argr{\textup{\texttt{arg2}}}
\newcommand\res{\textup{\texttt{res}}}
\newcommand\add{\textup{\texttt{add}}}
\newcommand\mult{\textup{\texttt{mult}}}
\newcommand\fbt{\textup{\texttt{fbt}}}
\newcommand\pset{\textup{\texttt{set}}}
\newcommand\phiop{\varphi_{({+},{\cdot})}}
\newcommand\Top{T_{({+},{\cdot})}}
\newcommand\phiset{\varphi_{\cont}}
\newcommand\Kset{\tsys_{\cont}}
\newcommand\xo{x}
\newcommand\yo{y}
\newcommand\zo{z}

\newcommand\ys{y}
\newcommand\xt{x}
\newcommand\yt{y}
\newcommand\ai[1]{a_{#1}}
\newcommand\ax{a}
\newcommand\pix[1]{\pi_{#1}}
\newcommand\phiz{\varphi_{(\cont,{+},{\cdot})}}
\newcommand\phizc{\phiz^{^{cl}}}
\newcommand\inter{\nu}
\newcommand\tsysI{\tsys_{\inter}}
\newcommand\tr [1]{\mathit{h}(#1)}

\definecolor{com}{RGB}{0,128,0}

\newcommand{\tileset}{\mathit{Ti}}
\newcommand{\tiling}{\mathit{ti}}

\renewcommand{\phi}{\varphi}
\newcommand{\east}{\mathit{east}}
\newcommand{\west}{\mathit{west}}
\newcommand{\north}{\mathit{north}}
\newcommand{\south}{\mathit{south}}

\newcommand\phicount{\phi^{\mathit{c}}}
\newcommand\phifb{\phi^{\mathit{fb}}}
\newcommand\tsysc{\tsys_0}

\newcommand{\prefs}[1]{\mathrm{Pref}(#1)}
\newcommand{\closure}[1]{\mathrm{cl}(#1)}
\newcommand{\phiopc}{\phiop^{cl}}
\newcommand{\diff}{D}

\begin{document}


\title[HyperLTL Is Highly Undecidable, HyperCTL* is Even Harder]{HyperLTL Satisfiability Is Highly Undecidable,\texorpdfstring{\\}{} HyperCTL* is Even Harder}

\titlecomment{{\lsuper*}A preliminary version of this article was presented at MFCS 2021~\cite{DBLP:conf/mfcs/FortinKT021}. Here, we extend this with new results on countable and finitely-branching satisfiability for \hyctlstar as well as a detailed study of the \hyltl alternation hierarchy.}

\author[M.~Fortin]{Marie Fortin\lmcsorcid{0000-0001-5278-0430}}[a]
\author[L.~Kuijer]{Louwe B. Kuijer\lmcsorcid{0000-0001-6696-9023}}[b]
\author[P.~Totzke]{Patrick Totzke\lmcsorcid{0000-0001-5274-8190}}[b]
\author[M.~Zimmermann]{Martin Zimmermann\lmcsorcid{0000-0002-8038-2453}}[c]

% affiliations are numbered automatically with a, b, c (see below)
% use the optional argument to indicate the affiliation(s) of each author
% omit the argument if there is only one author, or only one affiliation
\address{Université Paris Cité, CNRS, IRIF, France}
\email{mfortin@irif.fr}
\address{University of Liverpool, UK}
\email{Louwe.Kuijer@liverpool.ac.uk, totzke@liverpool.ac.uk}
\address{Aalborg University, Denmark}
\email{mzi@cs.aau.uk}

\begin{abstract}
Temporal logics for the specification of information-flow properties are able to express relations between multiple executions of a system.
The two most important such logics are HyperLTL and HyperCTL*, which generalise LTL and CTL* by trace quantification. 
It is known that this expressiveness comes at a price, i.e.\ satisfiability is undecidable for both logics. 

In this paper we settle the exact complexity of these problems, showing that both are in fact highly undecidable:
we prove that HyperLTL satisfiability is $\Sigma_1^1$-complete and HyperCTL* satisfiability is $\Sigma_1^2$-complete. 
These are significant increases over the previously known lower bounds and the first upper bounds.
To prove $\Sigma_1^2$-membership for HyperCTL*, we prove that every satisfiable HyperCTL* sentence has a model that is equinumerous to the continuum, the first upper bound of this kind. We also prove this bound to be tight.
Furthermore, we prove that both countable and finitely-branching satisfiability for HyperCTL* are as hard as truth in second-order arithmetic, i.e.\ still highly undecidable.

Finally, we show that the membership problem for every level of the HyperLTL quantifier alternation hierarchy is $\Pi_1^1$-complete.

\end{abstract}

\maketitle

%^^^^^^^^^^^^^^^^^^^^^^^^^^^^^^^^^^^^^^^^^^^^^^
%^^^^^^^^^^^^^^^^^^^^^^^^^^^^^^^^^^^^^^^^^^^^^^
\section{Introduction}
\label{sec:intro}

Most classical temporal logics like \ltl and \ctlstar refer to a single execution trace at a time while information-flow properties, which are crucial for security-critical systems,
require reasoning about multiple executions of a system. 
Clarkson and Schneider~\cite{ClarksonS10} coined the term \emph{hyperproperties} for such properties which, structurally, are sets \emph{of sets} of traces.
Just like ordinary trace and branching-time properties, hyperproperties can be specified using temporal logics, e.g.\ \hyltl and \hyctlstar~\cite{ClarksonFKMRS14}, expressive, but intuitive specification languages that are able to express typical information-flow properties such as noninterference, noninference, declassification, and input determinism.
Due to their practical relevance and theoretical elegance, hyperproperties and their specification languages have received considerable attention during the last decade.

\hyltl is obtained by extending \ltl~\cite{Pnueli77}, the most influential specification language for linear-time properties, by trace quantifiers to refer to multiple executions of a system. 
For example, the \hyltl formula
\[\forall \pi, \pi'.\ \G( i_\pi \leftrightarrow i_{\pi'}) \rightarrow \G (o_\pi \leftrightarrow o_{\pi'}) \]
expresses input determinism, i.e.\ every pair of traces that always has the same input (represented by the proposition~$i$) also always has the same output (represented by the proposition~$o$). 
Similarly, \hyctlstar is the extension of the branching-time logic \ctlstar~\cite{EmersonH86} by path quantifiers. 
\hyltl only allows formulas in prenex normal form while \hyctlstar allows arbitrary quantification, in particular under the scope of temporal operators. 
Consequently, \hyltl formulas are evaluated over sets of traces while \hyctlstar formulas are evaluated over transition systems, which yield the underlying branching structure of the traces. 

All basic verification problems, e.g.\ model checking~\cite{DBLP:conf/cav/BeutnerF22,DBLP:journals/corr/abs-2301-11229,Finkbeiner21,FinkbeinerRS15}, runtime monitoring~\cite{AgrawalB16,BonakdarpourF16,BrettSB17,DBLP:conf/csfw/BonakdarpourF18,DBLP:conf/atva/CoenenFHHS21,FinkbeinerHST18}, and synthesis~\cite{BonakdarpourF20,FinkbeinerHHT20,FinkbeinerHLST20}, have been studied. Most importantly, \hyctlstar model checking over finite transition systems is TOWER-complete~\cite{FinkbeinerRS15}, even for a fixed transition system~\cite{MZ20}. 
However, for a small number of alternations, efficient algorithms have been developed and were applied to a wide range of problems, e.g.\ an information-flow analysis of an I2C bus master~\cite{FinkbeinerRS15}, the symmetric access to a shared resource in a mutual exclusion protocol~\cite{FinkbeinerRS15}, and to detect the use of a defeat device to cheat in emission testing~\cite{BartheDFH16}.

But surprisingly, the exact complexity of the satisfiability problems for \hyltl and \hyctlstar is still open. 
Finkbeiner and Hahn proved that \hyltl satisfiability is undecidable~\cite{FinkbeinerH16}, a result which already holds when only considering finite sets of ultimately periodic traces
 and $\forall\exists$-formulas. 
In fact, Finkbeiner et al.\ showed that \hyltl satisfiability restricted to finite sets of ultimately periodic traces is $\Sigma_1^0$-complete~\cite{FHH18} (i.e.\ complete for the set of recursively enumerable problems).
Furthermore, Hahn and Finkbeiner proved that the $\exists^*\forall^*$-fragment has decidable satisfiability~\cite{FinkbeinerH16} while Mascle and Zimmermann studied the \hyltl satisfiability problem restricted to bounded sets of traces~\cite{MZ20}. The latter work implies that \hyltl satisfiability restricted to finite sets of traces (even non ultimately periodic ones) is also $\Sigma_1^0$-complete.
Following up on the results presented in the conference version of this article~\cite{DBLP:conf/mfcs/FortinKT021}, Beutner et al.\ studied satisfiability for safety and liveness fragments of \hyltl~\cite{DBLP:conf/lics/BeutnerCFHK22}.
Finally, Finkbeiner et al.\ developed tools and heuristics~\cite{DBLP:journals/corr/abs-2301-11229,FHH18,FinkbeinerHS17}.

As every \hyltl formula can be turned into an equisatisfiable \hyctlstar formula, \hyctlstar satisfiability is also undecidable. 
Moreover, Rabe has shown that it is even $\Sigma_1^1$-hard~\cite{Rabe16}, i.e.\ it is not even arithmetical.
However, both for \hyltl and for \hyctlstar satisfiability, only lower bounds, but no upper bounds, are known. 

\paragraph{Our Contributions.}
In this paper, we settle the complexity of the satisfiability problems for \hyltl and \hyctlstar
by determining exactly how undecidable they are.
That is, we provide matching lower and upper bounds in terms of the analytical hierarchy and beyond,
where decision problems (encoded as subsets of $\nats$) are classified based on their definability by formulas of higher-order arithmetic, namely by the type of objects one can quantify over and by the number of alternations of such quantifiers.
We refer to Roger's textbook~\cite{Rogers87} for fully formal definitions.
For our purposes, it suffices to recall the following classes.
$\Sigma_1^0$ contains the sets of natural numbers of the form
\[
\set{x \in \nats \mid \exists x_0.\  \cdots \exists x_k.\ \psi(x, x_0, \ldots, x_k)}
\] 
where quantifiers range over natural numbers and $\psi$ is a quantifier-free arithmetic formula.
The notation~$\Sigma_1^0$ signifies that there is a single block of existential quantifiers (the subscript~$1$) ranging over natural numbers (type~$0$ objects, explaining the superscript~$0$).
Analogously, $\Sigma_1^1$ is induced by arithmetic formulas with existential quantification of type~$1$ objects (functions mapping natural numbers to natural numbers) and arbitrary (universal and existential) quantification of type~$0$ objects.
Finally, $\Sigma_1^2$ is induced by arithmetic formulas with existential quantification of type~$2$ objects (functions mapping type~$1$ objects to natural numbers) and arbitrary  quantification of type~$0$ and type~$1$ objects.
So, $\Sigma_1^0$ is part of the first level of the arithmetic hierarchy, $\Sigma_1^1$ is part of the first level of the analytical hierarchy, while $\Sigma_1^2$ is not even analytical.

In terms of this classification, we prove that \hyltl satisfiability is $\Sigma_1^1$-complete while \hyctlstar satisfiability is $\Sigma_1^2$-complete, thereby settling the complexity of both problems and showing that they are highly undecidable.
In both cases, this is a significant increase of the lower bound and the first upper bound.

First, let us consider \hyltl satisfiability.
The $\Sigma_1^1$ lower bound is a straightforward reduction from the recurrent tiling problem, a standard $\Sigma_1^1$-complete problem asking whether $\nats\times\nats$ can be tiled by a given finite set of tiles.
So, let us consider the upper bound: $\Sigma_1^1$ allows to quantify over type~$1$ objects: functions from natural numbers to natural numbers, or, equivalently, over sets of natural numbers, i.e.\ countable objects.
On the other hand, \hyltl formulas are evaluated over sets of infinite traces, i.e.\ uncountable objects. 
Thus, to show that quantification over type~$1$ objects is sufficient, we need to apply a result of Finkbeiner and Zimmermann proving that every satisfiable \hyltl formula has a countable model~\cite{FZ17}.
Then, we can prove $\Sigma_1^1$-membership by expressing the existence of a model and the existence of appropriate Skolem functions for the trace quantifiers by type~$1$ quantification. 
We also prove that the satisfiability problem remains $\Sigma_1^1$-complete when restricted to ultimately periodic traces, or, equivalently, when restricted to finite traces.

Then, we turn our attention to \hyctlstar satisfiability.
Recall that \hyctlstar formulas are evaluated over (possibly infinite) transition systems, which can be much larger than type~$2$ objects, as the cardinality of type~$2$ objects is bounded by $\contcard$, the cardinality of the continuum.
Hence, to obtain our upper bound on the complexity we need, just like in the case of \hyltl, an upper bound on the size of minimal models of satisfiable \hyctlstar formulas.
To this end, we generalise the proof of Finkbeiner and Zimmermann to \hyctlstar, showing that every satisfiable \hyctlstar formula has a model of size~$\contcard$. 
We also exhibit a satisfiable \hyctlstar formula~$\varphi_\contcard$ whose models all have at least cardinality~$\contcard$, as they have to encode all subsets of $\nats$ by disjoint paths.
Thus, our upper bound~$\contcard$ is tight.

With this upper bound on the cardinality of models, we are able to prove $\Sigma_1^2$-membership of \hyctlstar satisfiability by expressing with type~$2$ quantification the existence of a model and the existence of a winning strategy in the induced model checking game.
The matching lower bound is proven by directly encoding the arithmetic formulas inducing $\Sigma_1^2$ as instances of the \hyctlstar satisfiability problem.
To this end, we use the formula~$\varphi_\contcard$ whose models have for each subset $A \subseteq \nats$ a path encoding $A$.
Now, quantification over type~$0$ objects (natural numbers) is simulated by quantification of a path encoding a singleton set, quantification over type~$1$ objects (which can be assumed to be sets of natural numbers) is simulated by quantification over the paths encoding such subsets, and existential quantification over type~$2$ objects (which can be assumed to be subsets of $\pow{\nats}$) is simulated by the choice of the model, i.e.\ a model encodes $k$ subsets of $\pow{\nats}$ if there are $k$ existential type~$2$ quantifiers.
Finally, the arithmetic operations can easily be implemented in \hyltl, and therefore also in \hyctlstar.

Using variations of these techniques, we also show that \hyctlstar satisfiability restricted to countable or to finitely branching models is equivalent to truth in second-order arithmetic, i.e.\ the question whether a given sentence of second-order arithmetic is satisfied in the structure~$(\nats, 0,1,+,\cdot,<)$. 
Restricting the class of models makes the problem simpler, but it is still highly-undecidable.

After settling the complexity of satisfiability, we turn our attention to the \hyltl quantifier alternation hierarchy and its relation to satisfiability. 
Rabe remarks that the hierarchy is strict~\cite{Rabe16}.
On the other hand, Mascle and Zimmermann show that every \hyltl formula has a polynomial-time computable equi-satisfiable formula with one quantifier alternation~\cite{MZ20}.
Here, we present a novel proof of strictness by embedding the \fol alternation hierarchy, which is also strict~\cite{CohenB71,Thomas81}. 
We use our construction to prove that for every $n > 0$, deciding whether a given formula is equivalent to a formula with at most $n$ quantifier alternations is $\Pi_1^1$-complete ($\Pi_1^1$ is the co-class of $\Sigma_1^1$, i.e.\ containing the complements of sets in $\Sigma_1^1$).

%^^^^^^^^^^^^^^^^^^^^^^^^^^^^^^^^^^^^^^^^^^^^^^
%^^^^^^^^^^^^^^^^^^^^^^^^^^^^^^^^^^^^^^^^^^^^^^
\section{Preliminaries}
\label{sec:definitions}
Fix a finite set~$\ap$ of atomic propositions. A \emph{trace} over $\ap$ is a map $t \colon \nats \rightarrow \pow{\ap}$, denoted by $t(0)t(1)t(2) \cdots$. 
It is \emph{ultimately periodic}, if $t = x \cdot y^\omega$ for some $x,y \in (\pow{\ap})^+$, i.e.\ there are $s,p>0$ with $t(n) = t(n+p)$ for all $n \ge s$.
The set of all traces over $\ap$ is  $(\pow{\ap})^\omega$. 

A transition system~$\tsys = (V,E,v_\initmark, \lambda)$ consists of a non-empty set~$V$ of vertices, a set~$E \subseteq V \times V$ of (directed) edges, an initial vertex~$v_\initmark \in V$, and a labelling~$\lambda\colon V \rightarrow \pow{\ap}$ of the vertices by sets of atomic propositions.
We require that each vertex has at least one outgoing edge. 
A path~$\rho$ through~$\tsys$ is an infinite sequence~$\rho(0)\rho(1)\rho(2)\cdots$ of vertices with $(\rho(n),\rho(n+1))\in E$ for every $n \ge 0$.
The trace of $\rho$ is defined as $\lambda(\rho(0))\lambda(\rho(1))\lambda(\rho(2))\cdots$.

\subsection{\hyltl}

The formulas of \hyltl are given by the grammar
\[
\phi  {} \cceq {}  \exists \pi.\ \phi \mid \forall \pi.\ \phi \mid \psi \qquad\qquad
\psi {}  \cceq {}  a_\pi \mid \neg \psi \mid \psi \vee \psi \mid \X \psi \mid \psi \U \psi
\]
where $a$ ranges over atomic propositions in $\ap$ and where $\pi$ ranges over a fixed countable set~$\var$ of \emph{(trace) variables}. Conjunction, implication, and equivalence are defined as usual, and the temporal operators eventually~$\F$ and always~$\G$ are derived as $\F\psi = \neg \psi\U \psi$ and $\G \psi = \neg \F \neg \psi$. A \emph{sentence} is a  formula without free variables.

The semantics of \hyltl is defined with respect to a \emph{trace assignment}, a partial mapping~$\Pi \colon \var \rightarrow (\pow{\ap})^\omega$. The assignment with empty domain is denoted by $\Pi_\emptyset$. Given a trace assignment~$\Pi$, a variable~$\pi$, and a trace~$t$ we denote by $\Pi[\pi \rightarrow t]$ the assignment that coincides with $\Pi$ everywhere but at $\pi$, which is mapped to $t$. 
Furthermore, $\suffix{\Pi}{j}$ denotes the trace assignment mapping every $\pi$ in $\Pi$'s domain to $\Pi(\pi)(j)\Pi(\pi)(j+1)\Pi(\pi)(j+2) \cdots $, its suffix from position $j$ onwards.

For sets~$T$ of traces and trace assignments~$\Pi$ we define 
\begin{itemize}
	\item $(T, \Pi) \models a_\pi$ if $a \in \Pi(\pi)(0)$,
	\item $(T, \Pi) \models \neg \psi$ if $(T, \Pi) \not\models \psi$,
	\item $(T, \Pi) \models \psi_1 \vee \psi_2 $ if $(T, \Pi) \models \psi_1$ or $(T, \Pi) \models \psi_2$,
	\item $(T, \Pi) \models \X \psi$ if $(T,\suffix{\Pi}{1}) \models \psi$,
	\item $(T, \Pi) \models \psi_1 \U \psi_2$ if there is a $j \ge 0$ such that $(T,\suffix{\Pi}{j}) \models \psi_2$ and for all $0 \le j' < j$: $(T,\suffix{\Pi}{j'}) \models \psi_1$, 
	\item $(T, \Pi) \models \exists \pi.\ \phi$ if there exists a trace~$t \in T$ such that $(T,\Pi[\pi \rightarrow t]) \models \phi$, and 
	\item $(T, \Pi) \models \forall \pi.\ \phi$ if for all traces~$t \in T$: $(T,\Pi[\pi \rightarrow t]) \models \phi$. 
\end{itemize}
We say that $T$ \emph{satisfies} a sentence~$\phi$ if $(T, \Pi_\emptyset) \models \phi$. In this case, we write $T \models \phi$ and say that $T$ is a \emph{model} of $\phi$. 
Two \hyltl sentences~$\varphi$ and $\varphi'$ are equivalent if $T \models \varphi$ if and only if $T \models \varphi'$ for every set~$T$ of traces.
Although \hyltl sentences are required to be in prenex normal form, they are closed under Boolean combinations, which can easily be seen by transforming such a formula into an equivalent formula in prenex normal form. 

\subsection{\texorpdfstring{\hyctlstar}{HyperCTL*}.}

The formulas of \hyctlstar are given by the grammar
\begin{align*}
\phi & {} \cceq {} a_\pi \mid \neg \phi \mid \phi \vee \phi \mid \X \phi \mid \phi \U \phi \mid \exists \pi.\ \phi \mid \forall \pi.\ \phi
\end{align*}
where $a$ ranges over atomic propositions in $\ap$ and where $\pi$ ranges over a fixed countable set~$\var$ of \emph{(path) variables}, and where we require that each temporal operator appears in the scope of a path quantifier. Again, other Boolean connectives and temporal operators are derived as usual.
Sentences are formulas without free variables.

Let $\tsys$ be a transition system. The semantics of \hyctlstar is defined with respect to a \emph{path assignment}, a partial mapping~$\Pi$ from variables in $\var$ to paths of $\tsys$. The assignment with empty domain is denoted by $\Pi_\emptyset$. Given a path assignment~$\Pi$, a variable~$\pi$, and a path~$\rho$ we denote by $\Pi[\pi \rightarrow \rho]$ the assignment that coincides with $\Pi$ everywhere but at $\pi$, which is mapped to $\rho$. 
Furthermore, $\suffix{\Pi}{j}$ denotes the path assignment mapping every $\pi$ in $\Pi$'s domain to $\Pi(\pi)(j)\Pi(\pi)(j+1)\Pi(\pi)(j+2) \cdots $, its suffix from position $j$ onwards.

 For transition systems~$\tsys$ and path assignments~$\Pi$ we define 
\begin{itemize}
	\item $(\tsys, \Pi) \models a_\pi$ if $a \in \lambda(\Pi(\pi)(0))$, where $\lambda$ is the labelling function of $\tsys$,
	\item $(\tsys, \Pi) \models \neg \psi$ if $(\tsys, \Pi) \not\models \psi$,
	\item $(\tsys, \Pi) \models \psi_1 \vee \psi_2 $ if $(\tsys, \Pi) \models \psi_1$ or $(\tsys, \Pi) \models \psi_2$,
	\item $(\tsys, \Pi) \models \X \psi$ if $(\tsys,\suffix{\Pi}{1}) \models \psi$,
	\item $(\tsys, \Pi) \models \psi_1 \U \psi_2$ if there exists a $j \ge 0$ such that $(\tsys,\suffix{\Pi}{j}) \models \psi_2$ and for all $0 \le j' < j$: $(\tsys,\suffix{\Pi}{j'}) \models \psi_1$, 
	\item $(\tsys, \Pi) \models \exists \pi.\ \phi$ if there exists a path~$\rho$ of $\tsys$, starting in $\last(\Pi)$, such that $(\tsys,\Pi[\pi \rightarrow \rho]) \models \phi$, and 
	\item $(\tsys, \Pi) \models \forall \pi.\ \phi$ if for all paths~$\rho$ of $\tsys$ starting in $\last(\Pi)$: $(\tsys,\Pi[\pi \rightarrow \rho]) \models \phi$. 
\end{itemize}
Here, $\last(\Pi)$ is the initial vertex of $\Pi(\pi)$, where $\pi$ is the path variable most recently added to or changed in $\Pi$, and the initial vertex of $\tsys$ if $\Pi$ is empty.\footnote{For the sake of simplicity, we refrain from formalising this notion properly, which would require to keep track of the order in which variables are added to or changed in $\Pi$.}
We say that $\tsys$ \emph{satisfies} a sentence~$\phi$ if $(\tsys, \Pi_\emptyset) \models \phi$. In this case, we write $\tsys \models \phi$ and say that $\tsys$ is a \emph{model} of $\phi$. 

\subsection{Complexity Classes for Undecidable Problems.}
A type~$0$ object is a natural number~$n \in \nats$, a type~$1$ object is a function $f \colon \nats \rightarrow \nats$, and a type~$2$ object is a function~$f\colon (\nats \rightarrow \nats) \rightarrow \nats$.
As usual, predicate logic with quantification over type~$0$ objects (first-order quantifiers) is called first-order logic.
Second- and third-order logic are defined similarly.

We consider formulas of arithmetic, i.e.\ predicate logic with signature~$(0,1, +, \cdot, <)$ evaluated over the natural numbers. With a single free variable of type~$0$, such formulas define sets of natural numbers (see, e.g.\ Rogers~\cite{Rogers87} for more details):
\begin{itemize}
	\item $\Sigma_1^0$ contains the sets of the form $\set{x \in \nats \mid \exists x_0.\  \cdots \exists x_k.\ \psi(x, x_0, \ldots, x_k)}$ where $\psi$ is a quantifier-free arithmetic formula and the $x_i$ are variables of type~$0$. 
	\item $\Sigma_1^1$ contains the sets of the form $\set{x \in \nats \mid \exists x_0.\  \cdots \exists x_k.\ \psi(x, x_0, \ldots, x_k)}$ where $\psi$ is an arithmetic formula with arbitrary (existential and universal) quantification over type~$0$ objects and the $x_i$ are variables of type~$1$. 
	\item $\Sigma_1^2$ contains the sets of the form $\set{x \in \nats \mid \exists x_0.\  \cdots \exists x_k.\ \psi(x, x_0, \ldots, x_k)}$ where $\psi$ is an arithmetic formula with arbitrary (existential and universal) quantification over type~$0$ and type~$1$ objects and the $x_i$ are variables of type~$2$.
\end{itemize}
Note that there is a bijection between functions of the form~$f\colon \nats\rightarrow \nats$ and subsets of $\nats$, which is implementable in  arithmetic. 
Similarly, there is a bijection between functions of the form~$f \colon (\nats \rightarrow \nats)\rightarrow\nats$ and subsets of $\pow{\nats}$, which is again implementable in arithmetic.
Thus, whenever convenient, we use quantification over sets of natural numbers and over sets of sets of natural numbers, instead of quantification over type~$1$ and type~$2$ objects; in particular when proving lower bounds. We then include $\in$ in the signature. 

Also, note that $0$ and $1$ are definable in first-order arithmetic. 
Thus, whenever convenient, we drop $0$ and $1$ from the signature of arithmetic.
In the same vein, every fixed natural number is definable in first-order arithmetic.

%^^^^^^^^^^^^^^^^^^^^^^^^^^^^^^^^^^^^^^^^^^^^^^
%^^^^^^^^^^^^^^^^^^^^^^^^^^^^^^^^^^^^^^^^^^^^^^
\section{\texorpdfstring{\hyltl}{HyperLTL} satisfiability is \texorpdfstring{$\Sigma_1^1$}{Sigma11}-complete}
\label{sec:sat}
In this section, we settle the complexity of the satisfiability problem for HyperLTL:
given a \hyltl sentence, determine whether it has a model.

\begin{thm}\label{thm:hyltl-sat}
  \hyltl satisfiability is $\Sigma^1_1$-complete.
\end{thm}

We should contrast this result with the $\Sigma^0_1$-completeness of
\hyltl satisfiability restricted to
\emph{finite} sets of ultimately periodic traces~\cite[Theorem 1]{FHH18}. The $\Sigma^1_1$-completeness
of \hyltl satisfiability in the general case implies that, in particular, the set of
satisfiable \hyltl sentences is neither recursively enumerable nor co-recursively
enumerable. A semi-decision procedure, like the one introduced in \cite{FHH18}
for finite sets of ultimately periodic traces, therefore cannot exist in general.

\subsection{\texorpdfstring{\hyltl}{HyperLTL} satisfiability is in \texorpdfstring{$\Sigma_1^1$}{Sigma11}}
The $\Sigma^1_1$ upper bound relies on the fact that every satisfiable \hyltl
formula has a countable model~\cite{FZ17}. 
This allows us to represent these models, and Skolem functions on them, by sets of natural numbers, which are 
type~$1$ objects. In this encoding, trace assignments are type~$0$ objects, as traces in a countable set can be identified by natural numbers. With some more existential type~$1$ quantification one can then express the existence of a function witnessing that every trace assignment consistent with the Skolem functions satisfies the quantifier-free part of the formula under consideration. 

\begin{lem}
\label{lemma:hyltl_membership}
\hyltl satisfiability is in $\Sigma_1^1$.	
\end{lem}

\begin{proof}
Let $\varphi$ be a \hyltl formula, let $\Phi$ denote the set of quantifier-free subformulas of~$\varphi$, and let $\Pi$ be a trace assignment whose domain contains the variables of $\varphi$. The expansion of $\varphi$ on $\Pi$ is the function~$\expansion_{\varphi, \Pi} \colon \Phi \times \nats \rightarrow \set{0,1}$ with
\[
\expansion_{\varphi, \Pi}(\psi, j)= \begin{cases}
1 &\text{if $\suffix{\Pi}{j} \models \psi$, and}\\
0 &\text{otherwise.}
\end{cases} 
\]
The expansion is completely characterised by the following consistency conditions:
\begin{itemize}
	\item $\expansion_{\varphi, \Pi} (a_\pi,j)=1$ if and only if $a \in \Pi(\pi)(j)$.
	\item $\expansion_{\varphi, \Pi} (\neg \psi,j)=1$ if and only if  $\expansion_{\varphi, \Pi} (\psi,j)=0$.
	\item $\expansion_{\varphi, \Pi} (\psi_1 \vee \psi_2,j)=1$ if and only if $\expansion_{\varphi, \Pi} (\psi_1,j)=1$ or $\expansion_{\varphi, \Pi} (\psi_2,j)=1$.
	\item $\expansion_{\varphi, \Pi} (\X\psi,j)=1$ if and only if $\expansion_{\varphi, \Pi} (\psi,j+1)=1$.
	\item $\expansion_{\varphi, \Pi} (\psi_1 \U \psi_2,j)=1$ if and only if there is a $j' \ge j$ such that $\expansion_{\varphi, \Pi} (\psi_2,j')=1$ and $\expansion_{\varphi, \Pi} (\psi_2,j'')=1$ for all $j''$ in the range~$j \le j'' < j'$.
\end{itemize}

Every satisfiable \hyltl sentence has a countable model~\cite{FZ17}.
Hence, to prove that the \hyltl satisfiability problem is in $\Sigma_1^1$, we express, for a given \hyltl sentence~$\varphi$ encoded as a natural number, the existence of the following type~$1$ objects (relying on the fact that there is a bijection between finite sequences over $\nats$ (denoted by $\nats^*$) and $\nats$ itself):
\begin{itemize}

	\item A countable set of traces over the propositions of $\varphi$ encoded as a function~$T$ from~$\nats \times \nats$ to $\nats$, mapping trace names and positions to (encodings of) subsets of the set of propositions appearing in $\varphi$.
	
	\item A function~$S$ from $\nats \times \nats^*$ to $\nats$ to be interpreted as Skolem functions for the existentially quantified variables of $\varphi$, i.e.\ we map a variable (identified by a natural number) and a trace assignment of the variables preceding it (encoded as a sequence of natural numbers) to a trace name.

	\item A function~$E$ from $\nats \times \nats \times \nats$ to $\nats$, where, for a fixed~$a \in\nats$ encoding a trace assignment~$\Pi$, the function~$x,y\mapsto E(a, x, y)$ is interpreted as the expansion of $\varphi$ on $\Pi$, i.e.\ $x$ encodes a subformula in $\Phi$ and $y$ is a position.

\end{itemize}
Then, we express the following properties using only type~$0$ quantification: For every trace assignment of the variables in $\varphi$, encoded by $a \in \nats$, if $a$ is consistent with the Skolem function encoded by $S$, then the function~$x,y\mapsto E(a, x, y)$ satisfies the consistency conditions characterising the expansion, and we have $E(a,x_0, 0) = 1$, where $x_0$ is the encoding of the maximal quantifier-free subformula of $\varphi$.
 We leave the tedious, but standard, details to the industrious reader. 
\end{proof}

\subsection{\texorpdfstring{\hyltl}{HyperLTL} satisfiability is \texorpdfstring{$\Sigma_1^1$}{Sigma11}-hard} 
To prove a matching lower bound, we reduce from the recurrent tiling problem~\cite{Harel85}, a standard $\Sigma_1^1$-complete problem.

\begin{lem}
\label{lemma:hyltl_hardness}
\hyltl satisfiability is $\Sigma_1^1$-hard.	
\end{lem}
\begin{proof}
A \emph{tile} is a function $\tau\colon\{\east,\west,\north,\south\}\to \mathcal{C}$
that maps directions into a finite set $\mathcal{C}$ of colours.
Given a finite set $\tileset$ of tiles, a \emph{tiling of the positive quadrant}
with $\tileset$ is a function $\tiling\colon\mathbb{N}\times\mathbb{N}\to \tileset$ with the property that:
\begin{itemize}
	\item if $\tiling(i,j)=\tau_1$ and $\tiling(i+1,j)=\tau_2$, then $\tau_1(\east)=\tau_2(\west)$ and
	\item if $\tiling(i,j)=\tau_1$ and $\tiling(i,j+1)=\tau_2$, then $\tau_1(\north)=\tau_2(\south)$.
\end{itemize}
The \emph{recurring tiling problem} is to determine, given a finite set $\tileset$ of tiles and a designated $\tau_0\in \tileset$, whether there is a tiling~$\tiling$ of the positive quadrant with $\tileset$ such that there are infinitely many $j\in\mathbb{N}$ such that $\tiling(0,j)=\tau_0$. 
This problem is known to be $\Sigma_1^1$-complete~\cite{Harel85}, so reducing it to \hyltl 
satisfiability will establish the desired hardness result. 

In our reduction, each $x$-coordinate in the positive quadrant will be represented by a trace, and each $y$-coordinate by a point in time.\footnote{Note that this means that if we were to visually represent this construction, traces would be arranged vertically.} In order to keep track of which trace represents which $x$-coordinate, we use one designated atomic proposition $x$ that holds on exactly one time point in each trace: $x$ holds at time $i$ if and only if the trace represents $x$-coordinate $i$.

For this purpose, let $\tileset$ and $\tau_0$ be given, and define the following formulas\footnote{Technically, the formula we define is not a \hyltl formula, since it is not in prenex normal form. Nevertheless, it can be trivially transformed into one by moving the quantifiers over conjunctions.} over $\ap = \set{ x } \cup \tileset$:

\begin{itemize}
  \item Every trace has exactly one point where $x$ holds:
  \[\varphi_1 = \forall \pi.\ (\neg x_\pi \U (x_\pi\wedge \X\G\neg x_\pi))\]
  \item   For every $i\in\mathbb{N}$, there is a trace with $x$ in the $i$-th position:
  \[\phi_2 = (\exists \pi.\ x_\pi) \wedge (\forall \pi_1.\ \exists \pi_2.\ \F(x_{\pi_1}\wedge \X x_{\pi_2}))\]
  \item
  If two traces represent the same $x$-coordinate, then they contain the same tiles:
  \[\varphi_3 = \forall \pi_1,\pi_2.\ (\F(x_{\pi_1}\wedge x_{\pi_2})\rightarrow \G(\bigwedge_{\tau\in \tileset}(\tau_{\pi_1}\leftrightarrow \tau_{\pi_2})))\]
  \item Every time point in every trace contains exactly one tile:
  \[\varphi_4=\forall \pi.\ \G\bigvee_{\tau\in \tileset}(\tau_\pi\wedge \bigwedge_{\tau'\in \tileset\setminus \{\tau\}}\neg (\tau')_\pi)\]
  \item Tiles match vertically:
  \[\varphi_5=\forall \pi.\ \G\bigvee_{\tau\in \tileset}(\tau_\pi\wedge \bigvee\nolimits_{\tau'\in\{\tau'\in \tileset\mid \tau(\north)=\tau'(\south)\}}\X (\tau')_\pi)\]
  \item Tiles match horizontally:
  \[\varphi_6=\forall \pi_1,\pi_2.\ (\F(x_{\pi_1}\wedge \X x_{\pi_2})\rightarrow \G\bigvee_{\tau\in \tileset}(\tau_{\pi_1}\wedge \bigvee\nolimits_{\tau'\in \{\tau'\in \tileset\mid \tau(\east)=\tau'(\west)\}}(\tau')_{\pi_2}))\]
  \item Tile~$\tau_0$ occurs infinitely often at $x$-position $0$:
  \[\varphi_7=\exists \pi.\ (x_\pi \wedge \G\F (\tau_0)_\pi)\]
\end{itemize}

Finally, take $\varphi_{\tileset} = \bigwedge_{1\leq n \leq 7}\varphi_n$. Collectively, subformulas $\varphi_1$--$\varphi_3$ are satisfied in exactly those sets of traces that can be interpreted as $\mathbb{N}\times\mathbb{N}$. Subformulas $\varphi_4$--$\varphi_6$ then hold if and only if the $\mathbb{N}\times\mathbb{N}$ grid is correctly tiled with $\tileset$. Subformula $\varphi_7$, finally, holds if and only if the tiling uses the tile $\tau_0$ infinitely often at $x$-coordinate $0$. Overall, this means $\varphi_{\tileset}$ is satisfiable if and only if $\tileset$ can recurrently tile the positive quadrant.

The $\Sigma^1_1$-hardness of \hyltl satisfiability therefore follows from the $\Sigma_1^1$-hardness of the recurring tiling problem~\cite{Harel85}.
\end{proof}

The $\Sigma^1_1$-completeness of \hyltl satisfiability still holds if we restrict to ultimately periodic traces.

\begin{thm}
\label{thm:hyltlsatup-copmleteness}
  \hyltl satisfiability restricted to sets of ultimately periodic traces is $\Sigma^1_1$-complete.
\end{thm}

\begin{proof}
The problem of whether there is a tiling of $\{(i,j)\in \mathbb{N}^2\mid i\geq j\}$, i.e.\ the part of $\mathbb{N}\times\mathbb{N}$ below the diagonal, such that a designated tile $\tau_0$ occurs on every row, is also $\Sigma_1^1$-complete~\cite{Harel85}.\footnote{The proof in \cite{Harel85} is for the part \emph{above} the diagonal with $\tau_0$ occurring on every column, but that is easily seen to be equivalent.} We reduce this problem to \hyltl satisfiability on ultimately periodic traces.

The reduction is very similar to the one discussed above, with the necessary changes being: (i) every time point beyond $x$ satisfies the special tile ``null'', (ii) horizontal and vertical matching are only checked at or before time point $x$ and (iii) for every trace~$\pi_1$ there is a trace~$\pi_2$ such that $\pi_2$ has designated tile~$\tau_0$ at the time where $\pi_1$ satisfies $x$ (so $\tau_0$ holds at least once in every row).

Membership in $\Sigma_1^1$ can be shown similarly to the proof of \autoref{lemma:hyltl_membership}. So, the problem is $\Sigma^1_1$-complete.
\end{proof}

Furthermore, a careful analysis of the proof of \autoref{thm:hyltlsatup-copmleteness} shows that we can restrict ourselves to ultimately periodic traces of the form~$x\cdot\emptyset^\omega$, i.e.\ to essentially finite traces.

Finally, let us also state for completeness the complexity of finite-state satisfiability, i.e., the question whether a given \hyltl sentence is satisfied by the set of traces of some finite transition system.
The result is a direct consequence of \hyltl model-checking being decidable~\cite{FinkbeinerRS15} (i.e., one can exhaustively enumerate and model-check all finite transition systems), yielding the upper bound, and a careful analysis of the undecidability proof for \hyltl satisfiability due to Finkbeiner and Hahn: they present a reduction from Post's correspondence problem (PCP) to \hyltl satisfiability so that the following are equivalent for any PCP instance~$P$ and resulting \hyltl sentence~$\varphi_P$:
\begin{itemize}
    \item $P$ has a solution.
    \item $\varphi_P$ is satisfiable.
    \item $\varphi_P$ has a finite model of ultimately periodic traces.
\end{itemize}

\begin{propC}[\cite{FinkbeinerRS15,FinkbeinerH16}]
\label{prop_hyltlfs}
\hyltl finite-state satisfiability is $\Sigma_1^0$-complete.
\end{propC}

%^^^^^^^^^^^^^^^^^^^^^^^^^^^^^^^^^^^^^^^^^^^^^^
%^^^^^^^^^^^^^^^^^^^^^^^^^^^^^^^^^^^^^^^^^^^^^^
\section{The \texorpdfstring{\hyltl}{HyperLTL} Quantifier Alternation Hierarchy}
\label{sec:hierarchy}
The number of quantifier alternations in
a formula is a crucial parameter in the complexity of \hyltl model-checking
\cite{FinkbeinerRS15,Rabe16}.
A natural question is then to understand \emph{which} properties
can be expressed with $n$ quantifier alternations, that is, given
a sentence~$\varphi$, determine if there exists an equivalent one with at
most $n$ alternations.
In this section, we show that this problem is in fact exactly as hard as the
\hyltl unsatisfiability problem (which asks whether a \hyltl sentence has no model),
and therefore $\Pi^1_1$-complete. Here, $\Pi_1^1$ is the co-class of $\Sigma_1^1$, i.e.\ it contains the complements of the $\Sigma_1^1$ sets.

\subsection{Definition and strictness of the hierarchy}

Formally, the \hyltl quantifier alternation hierarchy is defined as follows.
Let $\varphi$ be a \hyltl formula. We say that $\varphi$ is a $\Sigma_0$- or
a $\Pi_0$-formula if it is quantifier-free. It is a $\Sigma_n$-formula
if it is of the form~$\varphi = \exists \pi_1.\  \cdots \exists \pi_k.\ \psi$
and $\psi$ is a $\Pi_{n-1}$-formula. It is a $\Pi_n$-formula if
it is of the form~$\varphi = \forall \pi_1.\ \cdots \forall \pi_k.\ \psi$ and
$\psi$ is a $\Sigma_{n-1}$-formula.
We do not require each block of quantifiers to be non-empty, i.e.\ we may have
$k=0$ and $\varphi = \psi$.
Note that formulas in $\Sigma_0 = \Pi_0$ have free variables.
As we are only interested in sentences, we disregard $\Sigma_0 = \Pi_0$ in the following and only consider the levels~$\Sigma_n $ and $ \Pi_n$ for $n>0$.

By a slight abuse of notation, we also let $\Sigma_n$ denote the set of hyperproperties
definable by a $\Sigma_n$-sentence, that is, the set of all 
$L(\varphi) = \{ T \subseteq (2^\AP)^\omega \mid T \models \varphi\}$
such that $\varphi$ is a $\Sigma_n$-sentence of \hyltl.

\begin{thmC}[{\cite[Corollary 5.6.5]{Rabe16}}]
  The quantifier alternation hierarchy of \hyltl is strict:
  for all $n > 0$, $\Sigma_n \subsetneq \Sigma_{n+1}$.
\end{thmC}

The strictness of the hierarchy also holds if we restrict our attention to
sentences whose models consist of finite sets of traces that end in the suffix~$\emptyset^\omega$, i.e.\ that are essentially finite.

\begin{thm}\label{thm:strictness-finite}
  For all $n > 0$, there exists a $\Sigma_{n+1}$-sentence $\varphi$ of \hyltl
  that is not equivalent to any $\Sigma_{n}$-sentence, and such that for all
  $T \subseteq {(2^{\AP})}^\omega$, if $T \models \varphi$ then $T$ contains
  finitely many traces and $T \subseteq {(2^{\AP})}^\ast\emptyset^\omega$.
\end{thm}

This property 
is a necessary ingredient for our argument that membership
at some fixed level of the quantifier alternation hierarchy is $\Pi_1^1$-hard.
It could be derived from a small adaptation of the proof in~\cite{Rabe16},
and we provide for completeness an alternative proof by exhibiting a connection between the \hyltl quantifier alternation hierarchy and the quantifier alternation hierarchy for first-order logic over finite
words, which is known to be strict~\cite{CohenB71,Thomas82}.
The remainder of the subsection is dedicated to the proof of \autoref{thm:strictness-finite}.

The proof is organised as follows. We first define an encoding of finite words
as sets of traces. We then show that every first-order formula can be
translated into an equivalent (modulo encodings) \hyltl formula with the
same quantifier prefix (\autoref{lem:FO-to-hyltl}).
Finally, we show how to translate back \hyltl formulas into $\FOw$ formulas
with the same quantifier prefix (\autoref{lem:Qsimpl}), so that if the \hyltl quantifier alternation would
hierarchy collapse, then so would the hierarchy for $\FOw$.

\paragraph{First-Order Logic over Words.}
Let $\AP$ be a finite set of atomic propositions. A finite word over $\AP$
is a finite sequence $w = w(0)w(1) \cdots w(k)$ with $w(i) \in 2^\AP$ for all
$i$. We let $|w|$ denote the \emph{length} of $w$, and
$\pos(w) = \{0,\ldots,|w|-1\}$ the set of \emph{positions} of $w$.
The set of all finite words over $\AP$ is ${(2^\AP)}^\ast$.

Assume a countably infinite set of variables $\varfo$.
The set of $\FOw$ formulas is given by the grammar
\[
  \varphi ::= a(x) \mid x \le y \mid\lnot \varphi
  \mid \varphi \lor \varphi \mid  \exists x.\ \varphi \mid \forall x.\ \varphi \, ,
\]
where $a \in \AP$ and $x,y \in \varfo$.
The set of free variables of $\varphi$ is denoted $\Free(\varphi)$.
A sentence is a formula without free variables.

The semantics is defined as follows, $w \in {(2^\AP)}^\ast$ being a finite word
and $\nu : \Free(\varphi) \to \pos(w)$ an interpretation mapping variables to
positions in $w$:
\begin{itemize}
\item $(w, \nu) \models a(x)$ if $a \in w(\nu(x))$.
\item $(w, \nu) \models x \le y$ if $\nu(x) \le \nu(y)$.
\item $(w, \nu) \models \lnot \varphi$ if $w, \nu \not\models \varphi$.
\item $(w, \nu) \models \varphi \lor \psi$ if $w, \nu \models \varphi$ or
  $(w, \nu) \models \psi$.
\item $(w, \nu) \models \exists x.\ \varphi$ if there exists a position
  $n \in \pos(w)$ such that $(w,\nu[x \mapsto n]) \models \varphi$.
  \item $(w, \nu) \models \forall x.\ \varphi$ if for all positions
  $n \in \pos(w)$: $(w,\nu[x \mapsto n]) \models \varphi$.
\end{itemize}
If $\varphi$ is a sentence, we write
$w \models \varphi$ instead of $(w,\nu) \models \varphi$.

As for \hyltl, a $\FOw$ formula in prenex normal form is a $\Sigma_n$-formula if its quantifier
prefix consists of $n$ alternating blocks of
quantifiers (some of which may be empty), starting with a block of
existential quantifiers.
We let $\Sfo n$ denote the class of languages of finite words definable
by $\Sigma_n$-sentences.

\begin{thmC}[\cite{Thomas82,CohenB71}]\label{thm:strictness-FO}
  The quantifier alternation hierarchy of $\FOw$ is strict: for all
  $n \ge 0$, $\Sfo {n} \subsetneq \Sfo {n+1}$.
\end{thmC}

\paragraph{Encodings of Words}
The idea to prove \autoref{thm:strictness-finite} is to encode a word $w \in {(2^\AP)}^\ast$ as a set of traces
$T$ where each trace in $T$ corresponds to a position in $w$;
letters in the word are reflected in the label of the first
position of the corresponding trace in $T$, while the total order $<$ is
encoded using a fresh proposition $o \notin \ap$. More precisely, each trace
has a unique position labelled~$o$, distinct from one trace to another,
and traces are ordered according to the order of appearance of
the proposition $o$.
Note that there are several possible encodings for a same word, and we may fix
a canonical one when needed.
This is defined more formally below.

A \emph{stretch function} is a monotone funtion~$f: \Nat \to  \Nat\setminus\set{0}$, i.e.\ it satisfies $0 < f(0) < f(1) < \cdots$.
For all words $w \in {(2^{\AP})}^\ast$ and stretch functions~$f$, we define the set of traces~$\enc w f = \{t_n \mid n \in \pos(w)\} \subseteq (2^{\AP \cup \{o\}})^\ast\emptyset^\omega$
as follows: for all $i \in \Nat$,
\begin{itemize}
\item for all $a \in \AP$,
  $a \in t_n(i)$ if and only if $i = 0$ and $a \in w(n)$
\item $o \in t_n(i)$ if and only if $i = f(n)$.
\end{itemize}

It will be convenient to consider encodings with arbitrarily large spacing
between $o$'s positions.
To this end, for every $N \in \Nat$, we define a particular encoding
\[
  \encn w N = \enc w {n \mapsto N(n+1)} \, .
\]
So in $\encn w N$, two positions with non-empty labels are at distance at
least $N$ from one another.

Given $T = \enc w f$ and a trace assignment
$\Pi \colon \var \rightarrow T$,
we let $\tn T N = \encn w N$, and
$\nun \Pi N \colon \var \rightarrow \tn T N$ the trace assignment defined
by shifting the $o$ position in each $\Pi(\pi)$ accordingly, i.e.
\begin{itemize}
	\item $o \in \Pi^{(N)}(\pi) (N(i+1))$ if and only if $o \in \Pi(\pi)(f(i))$ and
	\item for all $a \in \ap$: $a \in \Pi^{N}(\pi)(0)$ if and only if $a \in \Pi(\pi)(0)$.
\end{itemize}

\paragraph{From FO to \hyltl}
We associate with every $\FOw$ formula $\varphi $ in prenex normal form
a \hyltl formula $\encPhi \varphi$ over $\ap \cup \{o\}$ by replacing
in $\varphi$:
\begin{itemize}
\item $a(x)$ with $a_x$, and
\item $x \le y$ with $\F(o_x \land \F o_y)$. 
\end{itemize}
In particular, $\encPhi \varphi$ has the same quantifier prefix as $\varphi$, which means that we treat variables of $\varphi$ as trace variables of $\encPhi{\varphi}$.

\begin{lem}\label{lem:FO-to-hyltl}
  For every $ \FOw$ sentence $\varphi $ in prenex normal form,
  $\varphi$ is equivalent to $\encPhi \varphi$ in the following sense:
  for all $w \in {(2^{\AP})}^\ast$ and all stretch functions~$f$,
  \[
    w \models \varphi \quad\text{if and only if}\quad
    \enc w f \models \encPhi \varphi \, .
  \]
\end{lem}

\begin{proof}
By induction over the construction of $\varphi$, relying on the fact that traces in $\enc w f$ are in bijection with positions in $w$.
\end{proof}

In particular, note that the evaluation of $\encPhi \varphi$ on $\enc w f$ does
not depend on $f$. We call such a formula \emph{stretch-invariant}:
a \hyltl sentence  $\varphi$ is \emph{stretch-invariant} if for all finite words $w$ and all stretch functions~$f$ and $g$,
\[
  \enc w f \models \varphi \quad\text{ if and only if }\quad
  \enc w g \models \varphi \, .
\]

\begin{lem}\label{lem:stretchinv}
  For all $\varphi \in \FOw$, $\encPhi \varphi$ is stretch-invariant.
\end{lem}

\begin{proof}
By induction over the construction of $\encPhi{\varphi}$, relying on the fact that the only temporal subformulas of $\encPhi{\varphi}$ are of the form~$\F(o_x \land \F o_y)$.
\end{proof}

\paragraph{Going Back From \hyltl to FO}
Let $\simpleHyperLTL$ denote the fragment of \hyltl consisting of
all formulas $\encPhi \varphi$, where $\varphi$ is an $\FOw$ formula in
prenex normal form.
Equivalently, $\psi \in \simpleHyperLTL$ if it is a \hyltl formula of the
form $\psi = Q_1x_1 \cdots Q_kx_k.\ \psi_0$, where $\psi_0$ is a Boolean
combination of formulas of the form  $a_x$ or $\F(o_x \land \F o_y)$.

Let us prove that every \hyltl sentence is equivalent, over
sets of traces of the form~$\enc w f$, to a sentence in $\simpleHyperLTL$
with the same quantifier prefix.
This means that if a \hyltl sentence $\encPhi \varphi$ is equivalent to a \hyltl
sentence with a smaller number of quantifier alternations, then it is also
equivalent over all word encodings to one of the form $\encPhi \psi$, which
in turns implies that the $\FOw$ sentences $\varphi$ and $\psi$ are equivalent.

The \emph{temporal depth} of a quantifier-free formula in \hyltl
is defined inductively as
\begin{itemize}\label{defdepth}
\item 	$\depth(a_\pi) = 0$,
\item  $\depth(\lnot \varphi) = \depth(\varphi)$,
\item $\depth(\varphi \lor \psi) = \max(\depth(\varphi),\depth(\psi))$,
\item $\depth(\X \varphi) = 1 + \depth(\varphi)$, and
\item $\depth(\varphi \U \psi) = 1 + \max(\depth(\varphi,\psi))$.
\end{itemize}
For a general \hyltl formula $\varphi = Q_1 \pi_1 \cdots Q_k \pi_k.\ \psi$,
we let $\depth(\varphi) = \depth(\psi)$.

\begin{lem}\label{lem:QFsimpl}
  Let $\psi$ be a quantifier-free formula of \hyltl.
  Let $N = \mathit{depth}(\psi)+1$.
  There exists a quantifier-free formula $\simpl \psi \in \simpleHyperLTL$
  such that for all $T = \enc w f$ and trace assignments $\Pi$,
  \[
    (\tn T N, \nun \Pi N) \models \psi \quad\text{if and only if}\quad (T,\Pi) \models \simpl \psi
     \, .
  \]
\end{lem}

\begin{proof} 
  Assume that $\Free(\psi) = \{\pi_1, \ldots, \pi_k\}$ is the set of free variables of $\psi$.
  Note that the value of $(\tn T N, \nun \Pi N) \models \psi$ depends
  only on the traces $\nun \Pi N(\pi_1),\ldots,\nun \Pi N(\pi_k)$.
  We see the tuple $(\nun \Pi N(\pi_1),\ldots,\nun \Pi N(\pi_k))$ as a single
  trace $w_{T,\Pi,N}$ over the set of propositions
  $\AP' = \{a_{\pi} \mid a \in \AP\cup\set{o} \land \pi \in \Free(\psi)\}$,
  and $\psi$ as an LTL formula over $\AP'$.
  
  We are going to show that the evaluation of $\psi$ over words $w_{T,\Pi,N}$ is
  entirely determined by the ordering of $o_{\pi_1}, \ldots, o_{\pi_n}$ in
  $w_{T,\Pi,N}$ and the label of $w_{T,\Pi,N}(0)$, which we can both
  describe using a formula in $\simpleHyperLTL$.
  The intuition is that non-empty labels in $w_{T,\Pi,N}$ are
  at distance at least $N$ from one another, and  a temporal formula of depth
  less than $N$ cannot distinguish between $w_{T,\Pi,N}$ and other words
  with the same sequence of non-empty labels and sufficient spacing between
  them.
  More generally, the following can be easily proved via
  Ehrenfeucht-Fra\"iss\'e games:

  \begin{clm}\label{claim:EF}
    Let $m,n \ge 0$, $(a_i)_{i \in \Nat}$ be a sequence of letters in
    $2^{\AP'}$, and
    \[
      w_1, w_2 \in
      \emptyset^m a_0 \emptyset^{n}\emptyset^\ast
      a_1 \emptyset^{n}\emptyset^\ast a_2 \emptyset^{n}\emptyset^\ast \cdots
    \]
    Then for all LTL formulas $\varphi$ such that $\depth(\varphi) \le n$,
    $w_1 \models \varphi$ if and only if $w_2 \models \varphi$.
  \end{clm}

  Here we are interested in words of a particular shape.
  Let $L_N$ be the set of infinite words $w \in {(2^{\AP'})}^\omega$ such that:
  \begin{itemize}
  \item For all $\pi \in \{\pi_1,\ldots,\pi_k\}$, there is a unique
    $i \in \Nat$ such that $o_{\pi} \in w(i)$. Moreover, $i \ge N$.
  \item If $o_{\pi} \in w(i)$ and $o_{\pi'} \in w(i') $,
    then $|i-i'| \ge N$ or $i = i'$.
  \item If $a_{\pi} \in w(i)$ for some $a \in \AP$ and
    $\pi \in \{\pi_1,\ldots,\pi_k\}$, then $i = 0$.
  \end{itemize}
  Notice that $w_{T,\Pi,N} \in L_N$ for all $T$ and all $\Pi$.
  
  For $w_1, w_2 \in L_N$, we write $w_1 \sim w_2$ if $w_1$ and $w_2$ differ
  only in the spacing between non-empty positions, that is, if there are
  $\ell \le k$ and $a_0, \ldots, a_\ell \in 2^{\AP'}$ such that
  $w_1,w_2 \in a_0 \emptyset^\ast a_1 \emptyset^\ast \cdots a_\ell \emptyset^\omega$.    
  Notice that $\sim$ is of finite index.
  Moreover, we can distinguish between its equivalence classes using formulas
  defined as follows. For all
  $A \subseteq \{a_{\pi} \mid a \in \AP \land \pi \in \{\pi_1,\ldots,\pi_k\}\}$
  and all total preorders $\preceq$ over $\{\pi_1, \ldots, \pi_k\}$,\footnote{
    I.e.\ $\preceq$ is required to be transitive and for all
    $\pi,\pi' \in \{\pi_1, \ldots, \pi_k\}$, we have $\pi \preceq \pi'$
    or $\pi' \preceq \pi$ (or both)} we let
  \[
    \varphi_{A,\preceq} =
    \bigwedge_{a \in A} a \land \bigwedge_{a \notin A} \lnot a
    \land \bigwedge_{\pi_i \preceq \pi_j} \F(o_{\pi_i} \land \F o_{\pi_j})
    \, .
  \]
  Note that every word $w \in L_N$ satisfies exactly one formula
  $\varphi_{A,\preceq}$, and that all words in an equivalence class
  satisfy the same one.
  We denote by $L_{A,\preceq}$ the equivalence class of $L_N/{\sim}$
  consisting of words satisfying $\varphi_{A,\preceq}$.
  So we have $L_N = \biguplus L_{A,\preceq}$.

  Since $\psi$ is of depth less than $N$, by \autoref{claim:EF}
  (with $n = N-1$ and $m = 0$), for all $w_1 \sim w_2$ we have
  $w_1 \models \psi$ if and only if $w_2 \models \psi$.
  Now, define $\simpl \psi$ as the disjunction of all $\varphi_{A,\preceq}$
  such that $\psi$ is satisfied by elements in the class $L_{A,\preceq}$.
  Then $\simpl \psi \in \simpleHyperLTL$, and
  \[
    \text{for all } w \in L_N, \quad w \models \simpl \psi
    \text{ if and only if } w \models \psi \, .
  \]
  In particular, for every $T$ and every $\Pi$, we have
  $(\tn T N,\nun \Pi N) \models \simpl \psi$ if and only if 
  $(\tn T N,\nun \Pi N) \models \psi$.
  Since the preorder between propositions $o_\pi$ and the label of the initial
  position are the same in $(\tn T N,\nun \Pi N)$ and $(T,\Pi)$,
  we also have $(T,\Pi) \models \simpl \psi$ if and only if
  $(\tn T N,\nun \Pi N) \models \simpl \psi$.
  Therefore,
  \[
(\tn T N, \nun \Pi N) \models \psi \quad\text{if and only if}\quad     (T,\Pi) \models \simpl \psi
     \, . \qedhere
  \]
\end{proof}

For a quantified \hyltl sentence $\varphi = Q_1\pi_1 \cdots Q_k\pi_k.\ \psi$, we let
$\simpl \varphi = Q_1\pi_1 \ldots Q_k\pi_k.\ \simpl \psi$, where $\simpl \psi$ is
the formula obtained through \autoref{lem:QFsimpl}.

\begin{lem}\label{lem:Qsimpl}
  For all \hyltl formulas~$\varphi $,
  for all $T = \enc w f$ and trace assignments~$\Pi$,
  \[
 (\tn T N, \nun \Pi N) \models \varphi    \text{ if and only if }
    (T, \Pi) \models \simpl \varphi \, ,
  \]
  where $N = \mathit{depth}(\varphi)+1$.
\end{lem}
\begin{proof}
  We prove the result by induction. Quantifier-free formulas are covered by \autoref{lem:QFsimpl}. 
  We have
  \begin{align*}
    (T, \Pi) \models \exists \pi.\ \simpl \psi
    & \quad\Leftrightarrow\quad
      \exists t \in T \text{ such that }
      (T,\Pi[\pi \mapsto t]) \models \simpl \psi \\
    & \quad\Leftrightarrow\quad
      \exists t \in T \text{ such that }
      (\tn T N, \nun {(\Pi[\pi \mapsto t])} N)
      \models \psi && \text{(IH)} \\
    & \quad\Leftrightarrow\quad
        \exists t \in \tn T N \text{ such that }
      (\tn T N, \nun {\Pi} N[\pi \mapsto t])
      \models \psi \\
    & \quad\Leftrightarrow\quad
        (\tn T N, \nun \Pi N) \models \exists \pi.\ \psi \, ,
  \end{align*}
  and similarly,
  \begin{align*}
    (T, \Pi) \models \forall \pi.\ \simpl \psi
    & \quad\Leftrightarrow\quad
      \forall t \in T, \text{we have }
      (T,\Pi[\pi \mapsto t]) \models \simpl \psi \\
    & \quad\Leftrightarrow\quad
      \forall t \in T, \text{we have }
      (\tn T N, \nun {(\Pi[\pi \mapsto t])} N)
      \models \psi && \text{(IH)} \\
    & \quad\Leftrightarrow\quad
        \forall t \in \tn T N, \text{we have }
      (\tn T N, \nun \Pi N [\pi \mapsto t])
      \models \psi \\
    & \quad\Leftrightarrow\quad
        (\tn T N, \nun \Pi N) \models \forall \pi.\ \psi \, . && \qedhere
  \end{align*}
\end{proof}

As a consequence, we obtain the following equivalence.

\begin{lem}\label{lem:eqsimpl}
  For all stretch-invariant \hyltl sentences $\varphi$ and
  for all $T = \enc w f$,
  \[
    T \models \varphi \quad\text{if and only if}\quad
    T \models \simpl \varphi \, .
  \]
\end{lem}

\begin{proof}
  By definition of $\varphi$ being stretch-invariant, we have
  $T \models \varphi$ if and only if
  $\tn T N \models \varphi$, which by \autoref{lem:Qsimpl} is
  equivalent to $T \models \simpl \varphi$.
\end{proof}

We are now ready to prove the strictness of the \hyltl quantifier alternation
hierarchy.

\begin{proof}[Proof of \autoref{thm:strictness-finite}.]
  Suppose towards a contradiction that the hierarchy collapses at
  level~$n > 0$, i.e.\ every \hyltl $\Sigma_{n+1}$-sentence is equivalent to some
  $\Sigma_n$-sentence.
  Let us show that the $\FOw$ quantifier alternation hierarchy also
  collapses at level $n$, a contradiction with \autoref{thm:strictness-FO}.

  Fix a $\Sigma_{n+1}$-sentence $\varphi$ of $\FOw$.
  The \hyltl sentence $\encPhi \varphi$ has the same quantifier prefix as
  $\varphi$, i.e.\ is also a $\Sigma_{n+1}$-sentence.
  Due to the assumed hierarchy collapse, there exists a  \hyltl
  $\Sigma_n$-sentence $\psi$ that is equivalent to $\encPhi \varphi$,
  and is stretch-invariant by \autoref{lem:stretchinv}.
  Then the \hyltl sentence $\simpl \psi$ defined above is also a $\Sigma_n$-sentence.
  Moreover, since $\simpl \psi \in \simpleHyperLTL$, there exists
  a $\FOw$ sentence $\varphi'$ such that $\simpl \psi = \encPhi {\varphi'}$,
which has the same quantifier prefix as $\simpl \psi$, i.e.\
  $\varphi'$ is a $\Sigma_n$-sentence of $\FOw$.
  For all words $w \in (2^\AP)^\ast$, we now have
  \begin{align*}
    w \models \varphi
    & \quad\text{if and only if}\quad \enc w f \models \encPhi \varphi
    && \text{(\autoref{lem:FO-to-hyltl})} \\
    & \quad\text{if and only if}\quad \enc w f \models \psi
    && \text{(assumption)} \\
    & \quad\text{if and only if}\quad \enc w f \models \simpl \psi
    && (\text{\autoref{lem:eqsimpl} and \autoref{lem:stretchinv}}) \\
    & \quad\text{if and only if}\quad \enc w f \models \encPhi {\varphi'}
    && (\text{definition}) \\
    & \quad\text{if and only if}\quad w \models \varphi'
    && \text{(\autoref{lem:FO-to-hyltl})}
  \end{align*}
  for an arbitrary stretch function~$f$.
  Therefore, $\Sfo {n+1} = \Sfo n$,  yielding the desired contradiction.

  This proves not only that for all $n > 0$, there is a \hyltl $\Sigma_{n+1}$-sentence
  that is not equivalent to any $\Sigma_n$-sentence,
  but also that there is one of the form $\encPhi \varphi$.
  Now, the proof still goes through if we replace $\encPhi \varphi$ by
  any formula equivalent to $\encPhi \varphi$ over all $\enc{w}{f}$,
  and in particular if we replace $\encPhi \varphi$ by
  $\encPhi \varphi \land \psi$, where the sentence
  \[
    \psi = \exists \pi.\ \forall \pi'.\
    (\F \G \emptyset_\pi) \land
    \G( \G \emptyset_\pi \rightarrow  \G \emptyset_{\pi'})  \]
    with $
    \emptyset_\pi = \bigwedge_{a \in \AP} \lnot a_\pi$
  selects models that contain finitely many traces, all in
  $(2^\AP)^\ast \cdot \emptyset^\omega$.
  Indeed, all $\enc w f$ satisfy~$\psi$.
  Notice that $\psi$ is a $\Sigma_2$-sentence, and since $n+1 \ge 2$,
  (the prenex normal form of) $\encPhi \varphi \land \psi$ is still a $\Sign {n+1}$-sentence.
\end{proof}

\subsection{Membership problem}

In this subsection, we investigate the complexity of the membership problem for the \hyltl quantifier alternation hierarchy.
Our goal is to prove the following result.

\begin{thm}\label{thm:Pi11}
  Fix $n > 0$. The problem of deciding whether a given \hyltl sentence is equivalent
  to some $\Sigma_n$-sentence is $\Pi^1_1$-complete.
\end{thm}

The easier part of the proof will be the upper bound, since a corollary of
\autoref{thm:hyltl-sat} is that the problem of deciding whether two
\hyltl formulas are equivalent is $\Pi^1_1$-complete.

The lower bound will be proven by a reduction from the \hyltl unsatisfiability problem.
The proof relies on \autoref{thm:strictness-finite}: given a sentence
$\varphi$, we are going to combine $\varphi$ with some $\Sigma_{n+1}$-sentence
$\varphi_{n+1}$ witnessing the strictness of the hierarchy, to construct a sentence $\psi$
such that $\varphi$ is unsatisfiable if and only if $\psi$ is equivalent to
a $\Sigma_n$-sentence.
Intuitively, the formula $\psi$ will describe models consisting of the
``disjoint union'' of a model of $\varphi_{n+1}$ and a model of~$\varphi$.
Here ``disjoint'' is to be understood in a strong sense: we split both
the set of traces and the time domain into two parts, used respectively to
encode the models of $\varphi_{n+1}$ and those of $\varphi$.

\begin{figure}
\centering
  \begin{tikzpicture}[yscale=0.6,font=\small]
    \node (03) at (0,3) {$\{a,b\}$};
    \node (13) at (1,3) {$\{a\}$};
    \node (23) at (2,3) {$\{a\}$};
    \node (02) at (0,2) {$\{b\}$};
    \node (12) at (1,2) {$\emptyset$};
    \node (22) at (2,2) {$\{a\}$};
    \begin{pgfonlayer}{background}
      \draw[teal,fill,fill opacity = 0.2,semithick]
      (03.north west) rectangle (8.5,1.6);
    \end{pgfonlayer}
    \node[teal] at (-1,2.5) {\large $T_\ell$} ;
    
    \node (31) at (3,1) {$\{a\}$};
    \node (41) at (4,1) {$\{a\}$};
    \node (51) at (5,1) {$\{a,b\}$};
    \node (61) at (6,1) {$\emptyset$};
    \node (71) at (7,1) {$\{a\}$};
    \node (81) at (8,1) {${\cdots}$};
    \node (30) at (3,0) {$\{b\}$};
    \node (40) at (4,0) {$\{a\}$};
    \node (50) at (5,0) {$\{b\}$};
    \node (60) at (6,0) {$\{a,b\}$};
    \node (70) at (7,0) {$\{a\}$};
    \node (80) at (8,0) {${\cdots}\vphantom{\{a\}}$};
    \begin{pgfonlayer}{background}
      \draw[brown,fill, fill opacity = 0.2,semithick]
      (31.north west) rectangle (8.5,-0.4) ;
      \node[brown] at (9,0.5) {\large $T_r$} ;
    \end{pgfonlayer}

    \foreach \i in {0,...,2} {
      \foreach \j in {0,1} {
        \node (ij) at (\i,\j) {$\{ \$ \}$} ;
      }
    }
    \foreach \i in {3,...,7} {
      \foreach \j in {2,3} {
        \node (ij) at (\i,\j) {$\{ \$ \}$} ;
      }
    }
    \node at (8,3) {${\cdots}$};
    \node at (8,2) {${\cdots}$};
  \end{tikzpicture}
  \caption{Example of a split set of traces where
  each row represents a trace and $b=3$.
      \label{fig:split}}
\end{figure}

\smallskip

To make this more precise, let us introduce some notations.
We assume a distinguished symbol $\$ \notin \AP$.
We say that a set of traces $T \subseteq {(2^{\AP\cup \{\$\}})}^\omega$ is
\emph{bounded} if there exists $b \in \Nat$ such that
$T \subseteq {(2^{\AP})}^b \cdot \{\$\}^\omega$.

\begin{lem}\label{lemma:bounded}
  There exists a  $\Pi_1$-sentence $\phib$ 
  such that for all $T \subseteq {(2^{\AP\cup \{\$\}})}^\omega$, we have
  $T \models \phib$ if and only if $T$ is bounded.
\end{lem}

\begin{proof}
  We let
    \begin{align*}
    \phib =
    \forall \pi, \pi'.\
    & (\lnot \$_\pi \U \G \$_\pi) \land
      \bigwedge_{a \in \AP} \G(\lnot (a_\pi \land \$_\pi)) \land
      \F \left(
      \lnot \$_\pi \land \lnot \$_{\pi'} \land \X \$_\pi \land \X \$_{\pi'}
      \right) \, .
    \end{align*}
    The conjunct $(\lnot \$_\pi \U \G \$_\pi) \land
    \bigwedge_{a \in \AP} \G(\lnot (a_\pi \land \$_\pi))$
    ensures that every trace is in ${(2^{\AP})}^\ast \cdot \{\$\}^\omega$,
    while
    $\F \left(
      \lnot \$_\pi \land \lnot \$_{\pi'} \land \X \$_\pi \land \X \$_{\pi'}
    \right)$
    ensures that the $\$$'s in any two traces $\pi$ and $\pi'$ start at the same
    position.
\end{proof}

We say that a nonempty set~$T$ of traces is \emph{split} if there exist a $b \in \Nat$ and $T_1$, $T_2$
such that $T = T_1 \uplus T_2$,
$T_1 \subseteq {(2^{\AP})}^b \cdot \{\$\}^\omega$, and
$T_2 \subseteq \{\$\}^b \cdot {(2^{\AP})}^\omega$.
Note that $b$ as well as $T_1$ and $T_2$ are unique then.
Hence, we define the left and right part of $T$ as $\Tl = T_1$ and
$\Tr  = \{t \in {(2^{\AP})}^\omega \mid \{\$\}^b \cdot t \in T_2\}$, respectively
(see \autoref{fig:split}).

It is easy to combine \hyltl specifications for the left and right part
of a split model into one global formula.

\begin{lem}
    \label{lem:dupsi}
  For all \hyltl sentences~$\varphi_\ell, \varphi_r$, one can construct
  a sentence $\psi$ such that for all split 
  $T \subseteq {(2^{\AP\cup \{\$\}})}^\omega$,
  it holds that   $\Tl \models \varphi_\ell$ and $\Tr  \models \varphi_r$ if and only if $T \models \psi$.
\end{lem}

\begin{proof}
  Let $\widehat {\varphi_r}$ denote the formula obtained from $\varphi_r$
  by replacing:
  \begin{itemize}
  \item every existential quantification $\exists \pi.\ \varphi$ with
    $\exists \pi.\ ((\F\G \lnot \$_\pi) \land \varphi)$;
  \item every universal quantification $\forall \pi.\ \varphi$ with
    $\forall \pi.\ ((\F\G \lnot \$_\pi) \rightarrow \varphi)$;
  \item the quantifier-free part $\varphi$ of $\varphi_r$ with
    $\$_\pi \U (\lnot \$_\pi \land \varphi)$,
    where $\pi$ is some free variable in $\varphi$.
  \end{itemize}
 Here, the first two replacements restrict quantification to traces in the right part while the last one requires the formula to hold at the first position of the right part. 
  We define $\widehat {\varphi_\ell}$ by similarly relativizing quantifications
  in $\varphi_\ell$.
  The formula $\widehat {\varphi_\ell} \land \widehat {\varphi_r}$ can then
  be put back into prenex normal form to define $\psi$.
\end{proof}

Conversely, any \hyltl formula that only has split models can be decomposed
into a Boolean combination of formulas that only talk about the left or right
part of the model. This is formalised in the lemma below.

\begin{lem}\label{lem:split}
  For all \hyltl $\Sigma_n$-sentences $\varphi$
  there exists a finite family $(\varphi^i_\ell, \varphi^i_r)_i$ of $\Sigma_n$-sentences 
  such that for all 
  split $T \subseteq {(2^{\AP\cup \{\$\}})}^\omega$: 
  $T \models \varphi$ if and only if there is an $i$ with
  $\Tl \models \varphi^i_\ell$ and $\Tr  \models \varphi^i_r$.
\end{lem}

\begin{proof}
  To prove this result by induction, we need to strengthen the statement to make it dual and allow for formulas with free variables.
  We let $\Free(\varphi)$ denote the set of free variables of a formula
  $\varphi$. 
  We prove the following result, which implies \autoref{lem:split}.

\begin{clm}\label{lem:splitgen}
  For all \hyltl $\Sigma_n$-formulas (resp.\ $\Pin n$-formulas) $\varphi$,
  there exists a finite family of $\Sigma_n$-formulas
  (resp.\ $\Pi_n$-formulas)
  $(\varphi^i_\ell, \varphi^i_r)_{i}$
  such that for all $i$,
  $\Free(\varphi) = \Free(\varphi^i_\ell) \uplus \Free(\varphi^i_r)$,
  and for all split $T$ and $\Pi$: $(T, \Pi) \models \varphi$
  if and only if there exists $i$ such that
  \begin{itemize}
  \item For all $\pi \in \Free(\varphi)$, $\Pi(\pi) \in \Tl$ if and only if
    $\pi \in \Free(\varphi^i_{\ell})$
    (and thus $\Pi(\pi) \in T \setminus \Tl$ if and only if
    $\pi \in \Free(\varphi^i_{r})$).
  \item $(\Tl, \Pi) \models \varphi^i_\ell$;
  \item $(\Tr , \Pi') \models \varphi^i_r$,
    where $\Pi'$ maps every $\pi \in \Free(\varphi^i_r)$ to the trace in
    $\Tr $ corresponding to $\Pi(\pi)$ in $T$
    (i.e.\ $\Pi(\pi) = \set{\$}^b \cdot \Pi'(\pi)$ for some $b$).
  \end{itemize}
\end{clm}

  To simplify, we can assume that the partition of the free variables of
  $\varphi$ into a left and right part is fixed, i.e.\
  we take $\Vl \subseteq \Free(\varphi)$
  and $\Vr = \Free(\varphi) \setminus \Vl$, and we restrict our attention
  to split $T$ and $\Pi$ such that $\Pi(\Vl) \subseteq \Tl$
  and $\Pi(\Vr) \subseteq T \setminus \Tl$.
  The formulas $(\varphi^i_\ell, \varphi^i_r)_{i}$ we are looking
  for should then be such that $\Free(\varphi^i_\ell) = \Vl$ and
  $\Free(\varphi^i_r) = \Vr$.
  If we can define sets of formulas $(\varphi^i_\ell, \varphi^i_r)$ for each
  choice of $\Vl, \Vr$, then the general case is solved by taking the union
  of all of those.
  So we focus on a fixed $\Vl,\Vr$, and prove the result by induction on
  the quantifier depth of $\varphi$.
  
  \paragraph{Base case}
  If $\varphi$ is quantifier-free, then it can be seen as an LTL formula
  over the set of propositions
  $\{a_\pi, \$_\pi \mid \pi \in \Free(\varphi), a \in \AP\}$,
  and any split model of $\varphi$ consistent with $\Vl,\Vr$
  can be seen as a word in $\Sigma_\ell^\ast \cdot \Sigma_r^\omega$,
  where
  \begin{align*}
    \Sigmal & = \big\{ \alpha \cup \{\$_\pi \mid \pi \in \Vr\} \mid 
    \alpha \subseteq {\{a_\pi \mid \pi \in \Vl \land a \in \AP\}} \big\} \text{ and }\\
    \Sigmar & = \big\{ \alpha \cup \{\$_\pi \mid \pi \in \Vl\} \mid 
    \alpha \subseteq {\{a_\pi \mid \pi \in \Vr \land a \in \AP\}} \big\} \, .
  \end{align*}
  Note in particular that $\Sigmal \cap \Sigmar = \emptyset$.
  We can thus conclude by applying the following standard result of
  formal language theory:
  
  \begin{clm}\label{lem:splitLTL}
    Let $L \subseteq \Sigma_1^\ast \cdot \Sigma_2^\omega$, where
    $\Sigma_1 \cap \Sigma_2 = \emptyset$.
    If $L = L(\varphi)$ for some LTL formula $\varphi$, then there exists a finite family~$(\varphi^i_1, \varphi^i_2)_{i}$  of LTL formulas 
    such that $L = \bigcup_{1 \le i \le k} L(\varphi^i_1) \cdot L(\varphi^i_2)$
    and for all $i$, $L(\varphi^i_1) \subseteq \Sigma_1^\ast$
    and $L(\varphi^i_2) \subseteq \Sigma_2^\omega$.
  \end{clm}

  \begin{proof}
    A language is definable in LTL if and only if it is accepted by some
    counter-free automaton~\cite{DG08,Thomas81}.
    Let $\mathcal{A}$ be a counter-free automaton for $L$.
    For every state~$q$ in~$\mathcal{A}$, let
    \begin{align*}
      L^q_1 & = \{w \in \Sigma_1^\ast \mid q_0 \xrightarrow{w} q
              \text{ for some initial state $q_0$}\} \text{ and }\\
      L^q_2 & = \{w \in \Sigma_2^\omega \mid \text{there is an accepting run
              on $w$ starting from $q$}\} \, .
    \end{align*}
    We have $L = \bigcup_q L^q_1\cdot  L^q_2$.
    Moreover, $L^q_1$ and $L^q_2$ are still recognisable by counter-free
    automata, and therefore LTL definable.
  \end{proof}
  
  \paragraph{Case $\varphi = \exists \pi.\ \psi$}
  Let $(\psi^i_{\ell,1},\psi^i_{r,1})$ and $(\psi^i_{\ell,2},\psi^i_{r,2})$
  be the formulas constructed respectively for
  $(\psi,\Vl \cup \{\pi\},\Vr)$ and $(\psi,\Vl,\Vr \cup \{\pi\})$.
  We take the union of all $(\exists \pi.\ \psi^i_{\ell,1},\psi^i_{r,1})$
  and $(\psi^i_{\ell,2},\exists \pi.\ \psi^i_{r,2})$.

  \paragraph{Case $\varphi = \forall \pi.\ \psi$}
  Let $(\xi^i_{\ell},\xi^i_{r})_{1 \le i \le k}$ be the formulas obtained for
  $\exists \pi.\ \lnot \psi$.
  We have ${(T,\Pi) \models \varphi}$ if and only if for all $i$,
  $(\Tl,\Pi) \not\models \xi^i_{\ell}$ or $(\Tr ,\Pi') \not\models \xi^i_r$;
  or, equivalently, if there exists $h: \{1,\ldots,k\} \to  \{\ell,r\}$ such
  that $(\Tl,\Pi) \models \bigwedge_{h(i) = \ell} \lnot \xi^i_{\ell}$ and
  $(\Tr ,\Pi') \models \bigwedge_{h(i) = r} \lnot \xi^i_{r}$.
  Take the family $(\varphi^h_{\ell},\varphi^h_{r})_h$, where
  $\varphi^h_{\ell} = \bigwedge_{h(i) = \ell} \lnot \xi^i_{\ell}$ and
  $\varphi^h_{r} = \bigwedge_{h(i) = r} \lnot \xi^i_{r}$.
  Since $\varphi = \forall \pi.\ \psi$ is a $\Pi_n$-formula,
  the formula $\exists \pi.\ \lnot \psi$ and by induction all
  $\xi^i_{\ell}$ and $\xi^i_{r}$ are $\Sigma_n$-formulas.
  Then all $\lnot \xi^i_{r}$ are $\Pi_n$-formulas, and since
  $\Pi_n$-formulas are closed under conjunction
  (up to formula equivalence), all $\varphi^h_{\ell}$ and $\varphi^h_{r}$
  are $\Pi_n$-formulas as well.
\end{proof}

We are now ready to prove \autoref{thm:Pi11}. 

\begin{proof}[Proof of \autoref{thm:Pi11}]
  The upper bound is an easy consequence of \autoref{thm:hyltl-sat}: Given a \hyltl sentence~$\varphi$, we express the existence of a $\Sigma_n$-sentence~$\psi$ using first-order quantification and encode equivalence of $\psi$ and $\varphi$
via the formula~$(\lnot \varphi \land \psi) \lor (\varphi \land \lnot \psi)$, which is unsatisfiable if and only if $\varphi$ and $\psi$ are equivalent. Altogether, this shows membership in $\Pi_1^1$, as $\Pi_1^1$ is closed under existential first-order quantification (see, e.g.~\cite[Page 82]{hinman}).

  We prove the lower bound by reduction from the unsatisfiability problem for
  \hyltl. So given a \hyltl sentence $\varphi$, we want to
  construct $\psi$ such that $\varphi$ is unsatisfiable if and only if
  $\psi$ is equivalent to a $\Sigma_n$-sentence.

  \smallskip
  
  We first consider the case $n > 1$.
  Fix a $\Sigma_{n+1}$-sentence~$\varphi_{n+1}$ that is in not
  equivalent to any $\Sigma_n$-sentence, and such that every model of
  $\varphi_{n+1}$ is bounded.
  The existence of such a formula is a consequence of
  \autoref{thm:strictness-finite}:
  There we have constructed $\varphi_{n+1}$ such that it only has finite models where every trace in the model ends in $\emptyset^\omega$. 
  The construction can easily be adapted so that $\varphi_{n+1}$ only has finite models where every trace in the model ends in $\emptyset^\omega$, i.e., it is bounded.
  Now, by \autoref{lem:dupsi}, there exists a computable $\psi$ such that for all split models
  $T$, we have $T \models \psi$ if and only if $\Tl \models \varphi_{n+1}$
  and $\Tr  \models \varphi$.
  
  First, it is clear that if $\varphi$ is unsatisfiable, then $\psi$ is
  unsatisfiable as well, and thus equivalent to
  $\exists \pi.\ a_\pi \land \lnot a_\pi$, which is a $\Sigma_n$-sentence
  since $n \ge 1$.
  
  Conversely, suppose towards a contradiction that $\varphi$ is satisfiable
  and that $\psi$ is equivalent to some $\Sigma_n$-sentence.
  Let $(\psi^i_\ell,\psi^i_r)_{i}$ be the finite family of $\Sigma_n$-sentences given
  by \autoref{lem:split} for $\psi$.
  Fix a model $T_\varphi$ of $\varphi$.
  For a bounded $T$, we let $\overline T$ denote the unique split set of traces such that
  $\overline \Tl = T$ and $\overline \Tr  = T_\varphi$.
  For all $T$, we then have $T \models \varphi_{n+1}$ if and only if $T$ is
  bounded and $\overline T \models \psi$.
  Recall that the set of bounded models can be defined by a
  $\Pi_1$-sentence~$\phib$ (\autoref{lemma:bounded}), which is also
  a $\Sigma_n$-sentence since $n > 1$.
  We then have $T \models \varphi_{n+1}$ if and only if $T \models \phib$ and
  there exists $i$ such that $T \models \psi^i_\ell$ and
  $T_\varphi \models \psi^i_r$.
  So $\varphi_{n+1}$ is equivalent to
  \[
    \phib \land \bigvee\nolimits_{i \text{ with } T_{\varphi} \models \psi^i_r}  \psi^i_\ell \, ,
  \]
  which, since $\Sigma_n$-sentences are closed (up to logical equivalence)
  under conjunction and disjunction, is equivalent to a $\Sigma_n$-sentence.
  This contradicts the definition of $\varphi_{n+1}$.

  \smallskip

  We are left with the case $n = 1$.
  Similarly, we construct $\psi$ such that $\varphi$ is unsatisfiable
  if and only if $\psi$ is unsatisfiable, and if and only if
  $\psi$ is equivalent to a $\Sigma_1$-sentence.
  However, we do not need to use bounded or split models here.
  Every satisfiable $\Sigma_1$-sentence has a model with finitely
  many traces.
  Therefore, a simple way to construct $\psi$ so that it is not equivalent to any
  $\Sigma_1$-sentence (unless it is unsatisfiable) is to ensure that every model
  of $\psi$ contains infinitely many traces.

  Let $x \notin \AP$, and
  $T_\omega = \{ \emptyset^n \{x\} \emptyset^\omega \mid n \in \Nat \}$.
  As seen in the proof of \autoref{lemma:hyltl_hardness},
  $T_\omega$ is definable in \hyltl:  There is a sentence
  $\varphi_\omega$ such that $T\subseteq (\pow{\ap \cup\set{x}})^\omega$ is a model of  $\varphi_\omega$ if and only if
  $T = T_\omega$.
  By relativising quantifiers in $\varphi_\omega$ and $\varphi$ to traces
  with or without the atomic proposition~$x$, one can construct a \hyltl
  sentence $\psi$ such that $T \models \psi$ if and only if
  $T_\omega \subseteq T$ and $T \setminus T_\omega \models \varphi$.

  Again, if $\varphi$ is unsatisfiable then $\psi$ is unsatisfiable and
  therefore equivalent to $\exists \pi.\ a_\pi \land \lnot a_\pi$,
  a $\Sigma_1$-sentence.
  Conversely, all models of $\psi$ contain infinitely many traces and therefore,
  if $\psi$ is equivalent to a $\Sigma_1$-sentence then it is unsatisfiable,
  and so is $\varphi$.
\end{proof}

%^^^^^^^^^^^^^^^^^^^^^^^^^^^^^^^^^^^^^^^^^^^^^^
%^^^^^^^^^^^^^^^^^^^^^^^^^^^^^^^^^^^^^^^^^^^^^^
\section{\texorpdfstring{\hyctlstar}{HyperCTL*} satisfiability is \texorpdfstring{$\Sigma^2_1$}{Sigma21}-complete}
\label{sec:ctlsat}

Here, we consider the \hyctlstar satisfiability problem: given a \hyctlstar sentence, determine whether it has a model $\tsys$ (of arbitrary size). We prove that it is much harder than \hyltl satisfiability.
As a key step of the proof, which is interesting in its own right, we also prove that every satisfiable sentence
admits a model of cardinality at most $\cont$ (the cardinality of the continuum). Conversely, we exhibit
a satisfiable \hyctlstar sentence whose models are all of cardinality at least~$\cont$.

\begin{thm}\label{thm:hyctl-sat}
  \hyctlstar satisfiability is $\Sigma^2_1$-complete.
\end{thm}

The proof proceeds as follows: In \autoref{subsec:hyctlmodelub}, we present the upper bound on the size of models of \hyctlstar sentences, which is then used in \autoref{subsec:hyctlub} to prove that \hyctlstar membership is in $\Sigma^2_1$: Intuitively, the existence of a transition system of cardinality~$\cont$ and the fact that it is indeed a model can be expressed using existential quantification of type~$2$ objects.
Finally, in \autoref{subsec:hyctllb} we present the matching lower bounds.
To this end, we first show that there is a satisfiable \hyctlstar sentence whose models have to contain pairwise disjoint paths encoding all traces over some fixed $\ap$. Note that this formula has only models of cardinality~$\contcard$, thereby giving a matching lower bound to the upper bound on the size of models of \hyctlstar.
Using this model, we can reduce existential third-order arithmetic to \hyctlstar satisfiability.

\subsection{An Upper Bound on the Size of \texorpdfstring{\hyctlstar}{HyperCTL*} Models}
\label{subsec:hyctlmodelub}
Before we begin proving membership in $\Sigma^2_1$, we obtain a bound on the size of minimal models of satisfiable \hyctlstar sentences. For this, we use an argument based on Skolem functions, which is a transfinite generalisation of the proof that all satisfiable \hyltl sentences have a countable model~\cite{FZ17}.
Later, we complement this upper bound by a matching lower bound, which will be applied in the hardness proof.

In the following, we use $\omega$ and $\omega_1$ to denote the first infinite and the first uncountable ordinal, respectively, and write $\aleph_0$ and $\aleph_1$ for their cardinality. 
As $\cont$ is uncountable, we have $\aleph_1 \le \cont$.

\begin{lem}\label{lemma:hyctl-model}
Each satisfiable HyperCTL$^*$ sentence~$\phi$ has a model of size at most~$\cont$.
\end{lem}

The proof of \autoref{lemma:hyctl-model} uses a Skolem function to create a model. Before giving this proof, we should therefore first introduce Skolem functions for \hyctlstar.

Let $\phi$ be a \hyctlstar formula. A quantifier in $\phi$ occurs \emph{with polarity 0} if it occurs inside the scope of an even number of negations, and \emph{with polarity 1} if it occurs inside the scope of an odd number of negations. We then say that a quantifier \emph{occurs existentially} if it is an existential quantifier with polarity 0, or a universal quantifier with polarity 1. Otherwise the quantifier \emph{occurs universally}. A Skolem function will map choices for the universally occurring quantifiers to choices for the existentially occurring quantifiers.

For reasons of ease of notation, it is convenient to consider a single Skolem function for all existentially occurring quantifiers in a \hyctlstar formula $\phi$, so the output of the function is an $l$-tuple of paths, where $l$ is the number of existentially occurring quantifiers in~$\phi$. The input consists of a $k$-tuple of paths, where $k$ is the number of universally occurring quantifiers in $\phi$, plus an $l$-tuple of integers. The reason for these integers is that we need to keep track of the time point in which the existentially occurring quantifiers are invoked.

Consider, for example, a \hyctlstar formula of the form $\forall \pi_1.\ \G \exists \pi_2.\ \psi$. This formula states that for every path $\pi_1$, and for every future point $\pi_1(i)$ on that path, there is some $\pi_2$ starting in $\pi_1(i)$ satisfying $\psi$. So the choice of $\pi_2$ depends not only on $\pi_1$, but also on $i$. For each existentially occurring quantifier, we need one integer to represent this time point at which it is invoked. A \hyctlstar Skolem function for a formula $\phi$ on a transition system~$\tsys$ is therefore a function~$f$ of the form~$f\colon\mathit{paths}(\tsys)^k\times \mathbb{N}^l\rightarrow \mathit{paths}(\tsys)^l$, where $\mathit{paths}(\tsys)$ is the set of paths over $\tsys$, $k$ is the number of universally occurring quantifiers in $\phi$ and $l$ is the number of existentially occurring quantifiers. 
Note that not every function of this form is a Skolem function, but for our upper bound it suffices that every Skolem function is of that form.

Now, we are able to prove that every satisfiable \hyctlstar formula has a  model of size~$\cont$.

\begin{proof}[Proof of \autoref{lemma:hyctl-model}]
If $\phi$ is satisfiable, let $\tsys$ be one of its models, and let $f$ be a Skolem function witnessing the satisfaction of $\phi$ on $\tsys$. We create a sequence of transition systems $\tsys_\alpha$ as follows.
\begin{itemize}
	\item $\tsys_0$ contains the vertices and edges of a single, arbitrarily chosen, path of $\tsys$ starting in the initial vertex.
	\item $\tsys_{\alpha+1}$ contains exactly those vertices and edges from $\tsys$ that are (i) part of $\tsys_{\alpha}$ or (ii) among the outputs of the Skolem function $f$ when restricted to input paths from $\tsys_\alpha$.
	\item if $\alpha$ is a limit ordinal, then $\tsys_\alpha= \bigcup_{\alpha'<\alpha} \tsys_{\alpha'}$.
\end{itemize}
Note that if $\alpha$ is a limit ordinal then $\tsys_\alpha$ may contain paths~$\rho(0)\rho(1)\rho(2)\cdots$ that are not included in any $\tsys_{\alpha'}$ with $\alpha' <\alpha$, as long as each finite prefix $\rho(0) \cdots \rho(i)$ is included in some $\alpha'_i<\alpha$.

First, we show that this procedure reaches a fixed point at $\alpha = \omega_1$. Suppose towards a contradiction that $\tsys_{\omega_1+1}\not =\tsys_{\omega_1}$. Then there are $\vec{\rho}=(\rho_1,\ldots,\rho_k)\in \mathit{paths}(\tsys_{\omega_1})^k$ and $\vec{n}\in \mathbb{N}^l$ such that $f(\vec{\rho},\vec{n})\not \in \mathit{paths}(\tsys_{\omega_1})^l$. Then for every $i\in\mathbb{N}$ and every $1\leq j \leq k$, there is an ordinal $\alpha_{i,j}<\omega_1$ such that the finite prefix $\rho_j(0)\cdots \rho_j(i)$ is contained in $\tsys_{\alpha_{i,j}}$. The set $\{\alpha_{i,j}\mid i\in \mathbb{N}, 1\leq j \leq k\}$ is countable, and because $\alpha_{i,j}<\omega_1$ each $\alpha_{i,j}$ is also countable. A countable union of countable sets is itself countable, so $\sup \{\alpha_{i,j}\mid i\in \mathbb{N}, 1 \leq j \leq k\}=\bigcup_{i\in\mathbb{N}}\bigcup_{1\leq j \leq k}\alpha_{i,j}=\beta< \omega_1$. 

But then the $\vec{\rho}$ are all contained in $\tsys_\beta$, and therefore $f(\vec{\rho},\vec{n})\in \mathit{paths}(\tsys_{\beta+1})^l$. But $\beta+1< \omega_1$, so this contradicts the assumption that $f(\vec{\rho},\vec{n})\not \in \mathit{paths}(\tsys_{\omega_1})^l$. From this contradiction we obtain $\tsys_{\omega_1+1}=\tsys_{\omega_1}$, so we have reached a fixed point. Furthermore, because $\tsys_{\omega_1}$ is contained in $\tsys$ and closed under the Skolem function and $\tsys$ satisfies $\phi$, we obtain that $\tsys_{\omega_1}$ also satisfies $\phi$.

Left to do, then, is to bound the size of $\tsys_{\omega_1}$, by bounding the number of vertices that get added at each step in its construction. We show by induction that $\size{\tsys_\alpha}\leq \cont$ for every $\alpha$. So, in particular, we have $\tsys_{\omega_1} \le \cont$, as required.

As base case, we have $\size{\tsys_{0}}
\le \aleph_0 < \cont$, since it consists the vertices of a single path. 
Consider then $\size{\tsys_{\alpha+1}}$. For each possible input to $f$, there are at most $l$ new paths, and therefore at most $\size{\nats \times l}$ new vertices in $\tsys_{\alpha+1}$. Further, there are $\size{\mathit{paths}(\tsys_\alpha)}^k\times \size{\mathbb{N}}^l$ such inputs. By the induction hypothesis, $\size{\tsys_\alpha}\leq \cont$, which implies that $\size{\mathit{paths}(\tsys_\alpha)}\leq \cont$. As such, the number of added vertices in each step is limited to $\cont^k\times \aleph_0^l \times \aleph_0\times l= \cont$. So $\size{\tsys_{\alpha+1}}\leq \size{\tsys_\alpha}+\cont=\cont$.

If $\alpha$ is a limit ordinal, $\tsys_\alpha$ is a union of at most $\aleph_1$ sets, each of which has, by the induction hypothesis, a size of at most $\cont$. Hence $\size{\tsys_\alpha}\leq \aleph_1\times \cont = \cont$.
\end{proof}

\subsection{\texorpdfstring{\hyctlstar}{HyperCTL*} satisfiability is in \texorpdfstring{$\Sigma^2_1$}{Sigma21}}
\label{subsec:hyctlub}

With the upper bound on the size of models at hand, we can place \hyctlstar satisfiability in $\Sigma_1^2$, as the existence of a model of size~$\cont$ can be captured by quantification over type~$2$ objects.

\begin{lem}\label{lemma:hyctl-membership}
\hyctlstar satisfiability is in $\Sigma_1^2$.
\end{lem}
\begin{proof}
As every \hyctlstar formula is satisfied in a model of size at most $\cont$, these models can be represented by objects of type~$2$. Checking whether a formula is satisfied in a transition system is equivalent to the existence of a winning strategy for Verifier in the induced model checking game. Such a strategy is again a type~$2$ object, which is existentially quantified. Finally, whether it is winning can be expressed by quantification over individual elements and paths, which are objects of types~$0$ and $1$.
Checking the satisfiability of a \hyctlstar formula $\phi$ therefore amounts to existential third-order quantification (to choose a model and a strategy) followed by a second-order formula to verify that $\phi$ holds on the model (i.e.\ that the chosen strategy is winning). Hence \hyctlstar satisfiability is in $\Sigma^2_1$.

Formally, we encode the existence of a winning strategy for Verifier in the \hyctlstar model checking game~$\mcgame(\tsys, \varphi)$ induced by a transition system~$\tsys$ and a \hyctlstar sentence~$\varphi$. 
This game is played between Verifier and Falsifier, one of them aiming to prove that $\tsys \models \varphi$ and the other aiming to prove $\tsys \not\models \varphi$. It is played in a graph whose positions correspond to subformulas which they want to check (and suitable path assignments of the free variables): each vertex (say, representing a subformula $\psi$) belongs to one of the players who has to pick a successor, which represents a subformula of $\psi$. A play ends at an atomic proposition, at which point the winner can be determined.

Formally, a vertex of the game is of the form~$(\Pi, \psi,b)$ where $\Pi$ is a path assignment, $\psi$ is a subformula of $\varphi$, and $b \in \set{0,1}$ is a flag used to count (modulo two) the number of negations encountered along the play; the initial vertex is $(\Pi_\emptyset, \varphi,0)$.
Furthermore, for until-subformulas~$\psi$, we need auxiliary vertices of the form~$(\Pi,\psi, b, j)$ with $j \in \nats$.
The vertices of Verifier are 
\begin{itemize}
	\item of the form $(\Pi, \psi, 0)$ with $\psi = \psi_1 \vee \psi_2$, $\psi = \psi_1 \U \psi_2$, or $\psi = \exists \pi.\ \psi'$, 
	\item of the form $(\Pi, \forall \pi.\ \psi', 1)$, or
	\item of the form $(\Pi, \psi_1 \U \psi_2, 1, j)$.
\end{itemize}
The moves of the game are defined as follows:
\begin{itemize}
	
	\item A vertex~$(\Pi, a_\pi, b)$ is terminal. It is winning for Verifier if $b = 0$ and $a \in \lambda(\Pi(\pi)(0))$ or if $b=1$ and $a \notin \lambda(\Pi(\pi)(0))$, where $\lambda$ is the labelling function of $\tsys$.
	
	\item A vertex~$(\Pi, \neg \psi, b)$ has a unique successor~$(\Pi,\psi, b+1 \bmod 2)$.
	
	\item A vertex~$(\Pi, \psi_1\vee \psi_2, b)$ has two successors of the form~$(\Pi, \psi_i, b)$ for $i \in \set{1,2}$.
	
	\item A vertex~$(\Pi, \X \psi, b)$ has a unique successor~$(\suffix{\Pi}{1}, \psi, b)$.
	
	\item A vertex~$(\Pi, \psi_1\U \psi_2, b)$ has a successor~$(\Pi, \psi_1\U \psi_2, b,j)$ for every $j \in \nats$.
	
	\item A vertex~$(\Pi, \psi_1\U \psi_2, b,j)$ has the  successor~$(\suffix{\Pi}{j}, \psi_2, b)$ as well as for every $0 \le j' <j$ the successor~$(\suffix{\Pi}{j'}, \psi_1, b)$ .
	
	\item A vertex~$(\Pi, \exists \pi.\ \psi, b)$ has successors~$(\Pi[\pi\mapsto \rho],\psi, b)$ for every path~$\rho$ of $\tsys$ starting in $\last(\Pi)$.
	
	\item A vertex~$(\Pi, \forall \pi.\ \psi, b)$ has successors~$(\Pi[\pi\mapsto \rho],\psi, b)$ for every path~$\rho$ of $\tsys$ starting in $\last(\Pi)$.
\end{itemize}

A play of the model checking game is a finite path through the graph, starting at the initial vertex and ending at a terminal vertex. It is winning for Verifier if the terminal vertex is winning for her. 
Note that the length of a play is bounded by~$2d$, where $d$ is the depth\footnote{The depth is the maximal nesting of quantifiers, Boolean connectives, and temporal operators.} of $\varphi$, as the formula is simplified during at least every other move.

A strategy~$\sigma$ for Verifier is a function mapping each of her vertices~$v$  to some successor of~$v$. 
A play~$v_0 \cdots v_k$ is consistent with $\sigma$, if $v_{k'+1} = \sigma(v_{k'})$ for every $0 \le k' < k$ such that $v_{k'}$ is a vertex of Verifier.
A straightforward induction shows that Verifier has a winning strategy for $\mcgame(\tsys,\varphi)$ if and only if $\tsys\models\varphi$.
 
Recall that every satisfiable \hyctlstar sentence has a model of cardinality~$\contcard$~(see \autoref{lemma:hyctl-model}). 
Thus, to place \hyctlstar satisfiability in $\Sigma_1^2$, we express, for a given natural number encoding a \hyctlstar formula~$\varphi$, the existence of the following type~$2$ objects (using suitable encodings):
\begin{itemize}
	\item A transition system~$\tsys$ of cardinality~$\contcard$. 
	\item A function~$\sigma$ from $V$ to $V$, where $V$ is the set of vertices of $\mcgame(\tsys,\varphi)$. Note that a single vertex of $V$ is a type~$1$ object. 
\end{itemize}
Then, we express that $\sigma$ is a strategy for Verifier, which is easily expressible using quantification over type~$1$ objects. 
Thus, it remains to express that $\sigma$ is winning by stating that every play (a sequence of type~$1$ objects of bounded length) that is consistent with $\sigma$ ends in a terminal vertex that is winning for Verifier. 
Again, we leave the tedious, but standard, details to the reader.
\end{proof}

\subsection{\texorpdfstring{\hyctlstar}{HyperCTL*} satisfiability is \texorpdfstring{$\Sigma^2_1$}{Sigma21}-hard}
\label{subsec:hyctllb}

Next, we prove a matching lower bound.
We first describe a satisfiable \hyctlstar sentence $\phiset$ that does not have any model of
cardinality less than $\cont$ (more precisely, the initial vertex must have
uncountably many successors), thus matching the upper bound from
\autoref{lemma:hyctl-model}.
We construct $\phiset$ with one particular model~$\Kset$ in mind,
defined below, though it also has other models.

The idea is that we want all possible subsets of $A \subseteq \Nat$ to
be represented in $\Kset$ in the form of paths
$\rho_A$ such that $\rho_A(i)$ is labelled by $1$ if $i \in A$, and by $0$ otherwise.
By ensuring that the first vertices of these paths are pairwise distinct, we obtain the desired lower bound on the cardinality.
We express this in \hyctlstar as follows: First, we express that there is a part of the model (labelled by $\fbt$) where every reachable vertex has two successors, one labelled with $0$ and one labelled with $1$, i.e.\ the unravelling of this part contains the full binary tree.
Thus, this part has a path~$\rho_A$ 
as above for every subset~$A$, but their initial vertices are not necessarily distinct.
Hence, we also express that there is another part (labelled by $\pset$) that contains a copy of each path in the $\fbt$-part, and that these paths indeed start at distinct successors of the initial vertex.

\begin{figure}
    \centering
\def\dist{2.0cm}
\begin{tikzpicture}[->,
>=stealth', 
level/.style={sibling distance = 3cm/#1, level distance = \dist}, 
scale=0.7,
transform shape,
grow=left,thick]

\node (init) [state] {\phantom{0}}
child {
    node (t0) [state] {0} 
    child {
        node (t00) [state] {0} 
        child {edge from parent[dashed] }
        child {edge from parent[dashed] }
    }
    child {
        node (t01) [state] {1} 
        child {edge from parent[dashed] }
        child {edge from parent[dashed] }
    }
}
child {
    node (t1) [state] {1} 
    child {
        node (t10) [state] {0} 
        child {edge from parent[dashed] }
        child {edge from parent[dashed] }
    }
    child {
        node (t11) [state] {1} 
        child {edge from parent[dashed] }
        child {edge from parent[dashed] }
    }
}
;

{\color{red}
\node[state,right=6cm of t00] (s00) {0}; 
\node[state,right=1cm of s00] (s000) {0}; 
\node[state,right=1cm of s000] (s0000) {0}; 
\draw (init) -- (s00);
\draw (s00) -- (s000);
\draw (s000) -- (s0000);
\draw[dotted] (s0000) -- +(\dist,0);

\draw[dotted] (init) -- (2.5,1.33);
\draw[-,dotted] (2.5,1.33) -- +(5,0);
\draw[dotted] (init) -- (2.5,1.66);
\draw[-,dotted] (2.5,1.66) -- +(5,0);

\draw[dotted] (init) -- (2.5,0.25);
\draw[-,dotted] (2.5,0.25) -- +(5,0);

\draw[dotted] (init) -- (2.5,-0.25);
\draw[-,dotted] (2.5,-0.25) -- +(5,0);

\draw[dotted] (init) -- (2.5,-1.33);
\draw[-,dotted] (2.5,-1.33) -- +(5,0);
\draw[dotted] (init) -- (2.5,-1.66);
\draw[-,dotted] (2.5,-1.66) -- +(5,0);

\node[state,right=6cm of t01] (s01) {0}; 
\node[state,right=1cm of s01] (s010) {1}; 
\draw (init) -- (s01);
\draw (s01) -- (s010);
\draw[dotted] (s010) -- +(3.5cm,0);

\node[state,right=6cm of t10] (s10) {1}; 
\node[state,right=1cm of s10] (s100) {0}; 
\draw (init) -- (s10);
\draw (s10) -- (s100);
\draw[dotted] (s100) -- +(3.5,0);

\node[state,right=6cm of t11] (s11) {1}; 
\node[state,right=1cm of s11] (s111) {1}; 
\node[state,right=1cm of s111] (s1111) {1}; 
\draw (init) -- (s11);
\draw (s11) -- (s111);
\draw (s111) -- (s1111);
\draw[dotted] (s1111) -- +(\dist,0);

}
\draw[->] (0,1cm) -- (init);

\end{tikzpicture}
    \caption{A depiction of $\Kset$. Vertices in black (on the left including the initial vertex) are labelled by $\fbt$, those in red (on the right, excluding the initial vertex) are labelled by $\pset$.
}%

    \label{fig:phiset}
\end{figure}

We let $\Kset = (V_\cont,E_\cont,t_\varepsilon,\lambda_\cont)$ (see \autoref{fig:phiset}), where
\begin{itemize}
	\item $V_\cont = \{t_u \mid u \in \{0,1\}^\ast\}
      \cup \{s^i_A \mid i \in \Nat \land A \subseteq \Nat\}$,
      \item $E_\cont
    = \{(t_u,t_{u0}),(t_u,t_{u1}) \mid u \in \{0,1\}^\ast\} \cup {} 
      \{ (t_\varepsilon, s^0_A) \mid A \subseteq \Nat \} \cup
      \{ (s^i_A,s^{i+1}_A) \mid A \subseteq \Nat, i \in \Nat \}$,
      \item and the labelling~$\lambda_\cont$ is defined as 
      \begin{itemize}
      	\item $\lambda_\cont(t_\varepsilon)  = \{\fbt\}$
      	\item $\lambda_\cont(t_{u \cdot 0}) = \{\fbt,0\}$
      	\item $\lambda_\cont(t_{u \cdot 1}) = \{\fbt,1\}$, and 
		\item $\lambda_\cont(s^i_A) = \begin{cases}
    \{\pset,0\} & \text{if } i \notin A, \\
    \{\pset,1\} & \text{if } i \in A.
  \end{cases}$      	
      \end{itemize} 
\end{itemize}

\begin{lem}\label{lem:phiset}
There is a satisfiable \hyctlstar sentence~$\phiset$ that has only models of cardinality at least $\contcard$.
\end{lem}

\begin{proof}
  The formula $\phiset$ is defined as the conjunction of the formulas below:
  \begin{enumerate}
  \item The label of the initial vertex is $\{\fbt\}$ and the labels of non-initial vertices are $\set{\fbt, 0}$, $\set{\fbt, 1}$, $\set{\pset, 0}$, or $\set{\pset, 1}$:
    \[
      \forall \pi.\ ( \fbt_\pi \land \lnot 0_\pi \land \lnot 1_\pi \land
      \lnot \pset_\pi )\wedge \X \G \big(
      (\pset_\pi \leftrightarrow \lnot \fbt_\pi) \land 
      (0_\pi \leftrightarrow \lnot 1_\pi)      \big)
    \]

  \item All $\fbt$-labelled vertices have a successor with label $\{\fbt,0\}$
    and one with label $\{\fbt,1\}$, and all $\fbt$-labelled vertices that are additionally labelled by $0$ or $1$ have no $\pset$-labelled successor:
    \[
      \forall \pi.\ \G\big(\fbt_\pi \rightarrow (
      (\exists \pi_0.\ \X (\fbt_{\pi_0} \land 0_{\pi_0}))
      \land (\exists \pi_1.\ \X (\fbt_{\pi_1} \land  1_{\pi_1})) \land
      ( (0_\pi \vee 1_{\pi}) \rightarrow \forall \pi'.\ \X \fbt_{\pi'}))
      \big)
    \]
   \item From $\pset$-labeled vertices, only $\pset$-labeled vertices are reachable:  
   \[\forall \pi.\ \G (\pset \rightarrow \G \pset) \] 
    
  \item For every path of $\fbt$-labelled vertices starting at a successor
    of the initial vertex, there is a path of $\pset$-labelled vertices
    (also starting at a successor of the initial vertex) with the same
    $\{0,1\}$ labelling:
    \[
      \forall \pi.\  \big( (\X \fbt_\pi) \rightarrow
      \exists \pi'.\ \X( \pset_{\pi'} \land \G(0_\pi \leftrightarrow 0_{\pi'}))
      \big)
    \]
  \item Any two paths starting in the same $\pset$-labelled vertex have the same sequence of labels:
    \[
      \forall \pi.\ \G \big(\pset_\pi \rightarrow
      \forall \pi'.\ \G(0_\pi \leftrightarrow 0_{\pi'})
      \big) \, .
    \]
  \end{enumerate}

  It is easy to check that $\Kset \models \phiset$.
  Note however that it is not the only model of $\phiset$: for instance,
  some paths may be duplicated, or merged after some steps if their label
  sequences share a common suffix.
  So, consider an arbitrary transition system $\tsys = (V,E,v_\initmark, \lambda)$
  such that $\tsys \models \phiset$.
  By condition 2, for every set $A \subseteq \Nat$,  there is a path
  $\rho_A$ starting at a successor of $v_\initmark$
  such that $\lambda(\rho_A(i)) = \{\fbt,1\}$ if $i \in A$ and
  $\lambda(\rho_A(i)) = \{\fbt,0\}$ if $i \notin A$.
  Condition 3 implies that there is also a $\pset$-labelled path $\rho'_A$
  such that $\rho'_A$ starts at a successor of $v_\initmark$, and has the same
  $\{0,1\}$ labelling as $\rho_A$.
  Finally, by condition 4, if $A \neq B$ then $\rho_A'(0) \neq \rho_{B}'(0)$.
  
  So, the initial vertex has at least as many successors as there are subsets of $\nats$, i.e., at least $\cont$ many.
\end{proof}

Before moving to the proof that \hyctlstar satisfiability is $\Sigma^2_1$-hard,
we introduce one last auxiliary formula that will be used in the reduction,
showing that addition and multiplication can be defined in \hyctlstar,
and in fact even in \hyltl, as follows:
  Let $\AP = \{\argl,\argr,\res,\add,\mult\}$ and let $\Top$ be the set of all traces $t \in {(2^{\AP})}^\omega$ such that 
  \begin{itemize}
  \item there are unique $n_1,n_2,n_3 \in \nats$ with $\argl \in t(n_1)$,
    $\argr \in t(n_2)$, and $\res \in t(n_3)$, and
  \item either $\add \in t(n)$ and $\mult \notin t(n)$  for all $n$ and $n_1+n_2 = n_3$,
    or $\mult \in t(n)$ and $\add \notin t(n)$  for all $n$ and $n_1 \cdot n_2 = n_3$.
  \end{itemize}

\begin{lem}\label{lem:phiop}
  There is a \hyltl sentence $\phiop$ which has $\Top$ as unique model.
\end{lem}

\begin{proof}
Consider the conjunction of the following \hyltl sentences:
\begin{enumerate}

	\item For every trace~$t$ there are unique $n_1,n_2,n_3 \in \nats$ with $\argl \in t(n_1)$,
    $\argr \in t(n_2)$, and $\res \in t(n_3)$:
\[
		\forall \pi.\ \bigwedge_{a \in \set{\argl,\argr,\res}}(\neg a_\pi)\U (a_\pi \wedge \X\G\neg a_\pi)
\]

	\item Every trace~$t$ satisfies either $\add \in t(n)$ and $\mult\notin t(n)$ for all $n$ or $\mult \in t(n)$ and $\add \notin t(n)$ for all $n$:
	\[
	\forall \pi.\ \G(\add_\pi\wedge\neg\mult_\pi) \vee \G(\mult_\pi\wedge\neg\add_\pi)
	\]

\end{enumerate}
	In the following, we only consider traces satisfying these formulas, as all others are not part of a model. Thus, we will speak of addition traces (if $\add$ holds) and multiplication traces (if $\mult $ holds).
	Furthermore, every trace encodes two unique arguments (given by the positions $n_1$ and $n_2$ such that $\argl\in t(n_1)$ and $\argr\in t(n_2)$) and a unique result (the position~$n_3$ such that $\res\in t(n_3)$). 
	
Next, we need to express that all possible arguments are represented in a model, i.e.\ for every $n_1$ and every $n_2$ there are two traces~$t$ with $\argl \in t(n_1)$ and $\argr \in t(n_2)$, one addition trace and one multiplication trace.	
	We do so inductively. 
	\begin{enumerate}
\setcounter{enumi}{2}
	
	\item\label{itemcomplstart} 
     There are two traces with both arguments being zero (i.e.\ $\argl$ and $\argr$ hold in the first position), one for addition and one for multiplication:
\[\bigwedge_{a \in \set{\add,\mult}}\exists \pi.\ a_\pi \wedge \argl_\pi \wedge \argr_\pi\]

\item\label{itemcomplstep} Now, we express that for every trace, say encoding the arguments~$n_1$ and $n_2$, the argument combinations~$(n_1+1, n_2)$ and $(n_1, n_2+1)$ are also represented in the model, again both for addition and multiplication (here we rely on the fact that either $\add$ or $\mult$ holds at every position, as specified above):
\begin{align*}
\forall \pi.\ \exists\pi_1, \pi_2.\ \left( \bigwedge_{i \in\set{1,2}} \add_\pi \leftrightarrow \add_{\pi_i} \right)\wedge 
&\F(\argl_\pi \wedge \X\argl_{\pi_1}) \wedge \F(\argr_\pi \wedge \argr_{\pi_1}) \mathrel{\wedge}\\
&
\F(\argl_\pi \wedge \argl_{\pi_2}) \wedge \F(\argr_\pi \wedge \X\argr_{\pi_2})	
\end{align*}  
\end{enumerate}	
Every model of these formulas contains a trace representing each possible combination of arguments, both for addition and multiplication. 

To conclude, we need to express that the result in each trace is correct. We do so by capturing the inductive definition of addition in terms of repeated increments (which can be expressed by the next operator) and the inductive definition of multiplication in terms of repeated addition.
Formally, this is captured by the next formulas:	
	\begin{enumerate}
	  \setcounter{enumi}{4}
	\item\label{itemcorrfirst} For every trace~$t$: if $\set{\add,\argl} \subseteq t(0)$ then $\argr$ and $\res$ have to hold at the same position (this captures $0 + n = n$):
	\[
	\forall \pi.\ (\add_\pi\wedge \argl_\pi ) \rightarrow \F(\argr_\pi \wedge \res_\pi)
	\]
	\item For each trace~$t$ with $\add\in t(0)$, $\argl \in t(n_1)$,
    $\argr \in t(n_2)$, and $\res \in t(n_3)$ such that $n_1 > 0$ there is a trace~$t'$ such that $\add\in t'(0)$, $\argl \in t'(n_1-1)$,
    $\argr \in t'(n_2)$, and $\res \in t'(n_3-1)$ (this captures $n_1 + n_2 = n_3 \Leftrightarrow n_1-1 + n_2 = n_3-1$ for $n_1 > 0$):
    \\[.5em]$\displaystyle
    \forall \pi.\ \exists\pi'.\ (\add_\pi \wedge \neg \argl_\pi) \rightarrow\\ \phantom{x}\hfill \left(\add_{\pi'} \wedge  \F( \X\argl_{\pi} \wedge \argl_{\pi'} ) \wedge \F(\argr_\pi \wedge \argr_{\pi'}) \wedge \F(\X\res_{\pi} \wedge \res_{\pi'}) \right)$ \\[-.5em]
    \item For every trace~$t$: if $\set{\mult,\argl} \subseteq t(0)$ then also $\res \in t(0)$ (this captures $0 \cdot n = 0$):
	\[
	\forall \pi.\ (\mult_\pi\wedge \argl_\pi ) \rightarrow  \res_\pi
	\]
	\item\label{itemcorrlast} Similarly, for each trace~$t$ with $\mult \in t(0)$, $\argl \in t(n_1)$,
    $\argr \in t(n_2)$, and $\res \in t(n_3)$ such that $n_1 > 0$ there is a trace~$t'$ such that $\mult \in t'(0)$, $\argl \in t'(n_1-1)$,
    $\argr \in t'(n_2)$, and $\res \in t'(n_3-n_2)$.
   The latter requirement is expressed by the existence of a trace~$t''$ with $\add\in t''(0)$, $\argr \in t''(n_2)$, $\res \in t''(n_3)$, and $\argl$ holding in $t''$ at the same time as $\res$ in $t'$, which implies $\res \in t'(n_3-n_2)$. Altogether, this captures $n_1 \cdot n_2 = n_3 \Leftrightarrow (n_1-1) \cdot n_2 = n_3-n_2$ for $n_1 > 0$.
    \begin{align*}
    \forall \pi.\ \exists\pi',\pi''.\ &(\mult_\pi \wedge  \neg \argl_\pi) \rightarrow   \left(\mult_{\pi'} \wedge \add_{\pi''} \wedge\right.  \\ 
&    
    \F(\X\argl_{\pi} \wedge \argl_{\pi'} ) \wedge 
    \F(\argr_\pi \wedge \argr_{\pi'} \wedge \argr_{\pi''}) \wedge \\ 
&  \left.    \F(\res_{\pi'} \wedge \argl_{\pi''}) \wedge 
  \F(\res_{\pi} \wedge \res_{\pi''}
    )\right)
    \end{align*}
\end{enumerate}

Now, $\Top$ is a model of the conjunction~$\phiop$ of these eight formulas. Conversely, every model of $\phiop$ contains all possible combinations of arguments (both for addition and multiplication) due to Formulas~(\ref{itemcomplstart}) and (\ref{itemcomplstep}).
Now, Formulas~(\ref{itemcorrfirst}) to (\ref{itemcorrlast}) ensure that the result is \emph{correct} on these traces. Altogether, this implies that $\Top$ is the unique model of $\phiop$.
\end{proof}

To establish $\Sigma_1^2$-hardness, we give an encoding of formulas of
existential third-order arithmetic into \hyctlstar, i.e.\ every formula of the form~$\exists \xt_1.\ \ldots \exists \xt_n.\ \psi$ where
  $\xt_1, \ldots, \xt_n$ are third-order variables
  and $\psi$ is a sentence of second-order arithmetic can be translated into a \hyctlstar sentence.

As explained in \autoref{sec:definitions}, we can (and do for the remainder
of the section) assume that first-order (type 0) variables range over natural
numbers,
second-order (type 1) variables range over sets of natural numbers,
and third-order (type~$2$) variables range over sets of sets of natural numbers.

\begin{lem}\label{lem:hyctlstar-reduction}
One can effectively translate sentences~$\varphi$ of existential third-order arithmetic into \hyctlstar sentences~$\varphi'$ such that
  $(\Nat, +, \cdot, <,\in)$ is a model of $\varphi$ if and only if
  $\varphi'$ is satisfiable.
\end{lem}

\begin{proof}  
  The idea of the proof is as follows. We represent sets of natural numbers
  as infinite paths with labels in $\{0,1\}$, so that quantification over
  sets of natural numbers in $\psi$ can be replaced by \hyctlstar path
  quantification.
    First-order quantification is handled in the same way, but using
  paths where exactly one vertex is labelled $1$.
In particular we encode first- and second-order variables~$x$ of $\varphi$ as path variables~$\pix{x}$ of $\varphi'$.
  For this to work, we need to make sure that every possible set
  has a path representative in the transition system (possibly several
  isomorphic ones).
  This is where formula $\phiset$ defined in \autoref{lem:phiset} is used.
  For arithmetical operations, we rely on the formula $\phiop$ from
  \autoref{lem:phiop}.
  Finally, we associate with every existentially quantified third-order
  variable $\xt_i$ an atomic proposition $\ai i$, so that
  for a second-order variable $\ys$, the atomic formula~$\ys \in \xt_i$ is interpreted as the
  atomic proposition $\ai i$ being true on the second vertex of $\pix y$.
  This is all explained in more details below.
  \smallskip
  
  Let $\varphi = \exists \xt_1.\ \ldots \exists \xt_n.\ \psi$ where
  $\xt_1, \ldots, \xt_n$ are third-order variables
  and $\psi$ is a formula of second-order arithmetic. 
  We use the atomic propositions 
  \[\AP = \{a_1,\ldots,a_n,0,1,\pset,\fbt,\argl,\argr,\res,\mult,\add\}.\]
  Given an interpretation
  $\inter : \{\xt_1,\ldots,\xt_n\} \rightarrow 2^{(2^\Nat)}$
  of the third-order variables of $\varphi$, we denote by $\tsysI$ the
  transition system over $\AP$ obtained as follows:
  We start from $\Kset$, and extend it with an $\{a_1,\ldots,a_n\}$-labelling
  by setting $\ai i \in \lambda(\rho_A(0))$ if $A \in \inter(\xt_i)$;
  then, we add to this transition system all traces in $\Top$ as disjoint
  paths below the initial vertex.

  From the formulas $\phiset$ and $\phiop$ defined in Lemmas~\ref{lem:phiset}
  and \ref{lem:phiop}, it is not difficult to construct a formula
  $\phiz$ such that:
  \begin{itemize}
  \item For all $\inter : \{\xt_1,\ldots,\xt_n\} \rightarrow 2^{(2^\Nat)}$,
    the transition system $\tsysI$ is a model of $\phiz$.
  \item Conversely, in any model $\tsys = (V,E,v_\initmark, \lambda)$ of $\phiz$, the following conditions
    are satisfied:
  \begin{enumerate}
    \item
      For every path $\rho$ starting at a $\pset$-labelled successor of
      the initial vertex $v_\initmark$, the vertex~$\rho(0)$ has a label of the form $\lambda(\rho(0)) = \{\pset,b\} \cup \ell$
      with $b \in \{0,1\}$ and $\ell \subseteq \{\ai 1, \ldots, \ai n\}$,
      and every vertex $\rho(i)$ with $i > 0$ has a label
      $\lambda(\rho(i)) = \{\pset,0\}$ or $\lambda(\rho(i)) = \{\pset,1\}$.
      
    \item For every $A \subseteq \Nat$, there exists a $\pset$-labelled path
      $\rho_A$ starting at a successor of $v_\initmark$ 
      such that $1 \in \lambda(\rho_A(i))$ if  $i \in A$,
      and $0 \in \lambda(\rho_A(i))$ if  $i \notin A$.
      Moreover, all such paths have the same $\{\ai 1,\ldots,\ai n\}$ labelling;
      this can be expressed by the formula
      \[
        \forall \pi, \pi'.\
        \X\left( \Big(\G (\pset_\pi \land \pset_{\pi'} \land
        (1_\pi \leftrightarrow 1_{\pi'}))\Big) \rightarrow
        \bigwedge\nolimits_{\ax \in \{\ai 1, \ldots, \ai n\}}
        \ax_\pi \leftrightarrow \ax_{\pi'}\right)
        \, .
      \]
    \item For every path $\rho$ starting at an $\add$- or $\mult$-labelled
      successor of the initial vertex, the label sequence
      $\lambda(\rho(0)) \lambda(\rho(1)) \cdots$ of $\rho$ is in $\Top$.

    \item Conversely, for every trace $t \in \Top$, there exists a path
      $\rho$ starting at a successor of the initial vertex such that
      $\lambda(\rho(0)) \lambda(\rho(1)) \cdots = t$.
    \end{enumerate}
  \end{itemize}

  We then let $\varphi' = \phiz  \land \tr  \psi$, where $\tr  \psi$ is defined inductively from the second-order body $\psi$
  of $\varphi$ as follows:
  \begin{itemize}
  \item $\tr {\psi_1 \lor \psi_2} = \tr {\psi_1} \lor \tr {\psi_2}$.
  \item $\tr {\lnot \psi_1} = \lnot \tr {\psi_1}$.
  \item If $x$ ranges over sets of natural numbers,
    \[\tr {\exists x.\ \psi_1} =
    \exists \pix x.\ ((\X \pset_{\pix x}) \land \tr  {\psi_1}),\]
    and
    \[\tr {\forall x.\ \psi_1} =
    \forall \pix x.\ ((\X \pset_{\pix x}) \rightarrow \tr  {\psi_1}).\]

  \item If $x$ ranges over natural numbers, 
    \[\tr {\exists x.\ \psi_1} = \exists \pix x.\ ((\X \pset_{\pix x}) \land
    \X (0_{\pix x} \U (1_{\pix x} \land \X\G 0_{\pix x})) \land \tr  {\psi_1}),\]
    and
    \[\tr {\forall x.\ \psi_1} = \forall \pix x.\ ((\X \pset_{\pix x}) \land
    \X (0_{\pix x} \U (1_{\pix x} \land \X\G 0_{\pix x})) \rightarrow \tr  {\psi_1}).\] 
    Here, the subformula~$0_{\pix x} \U (1_{\pix x} \land \X\G 0_{\pix x})$ expresses that there is a single $1$ on the trace assigned to $\pix x$, i.e.\ the path represents a singleton set.
  \item If $\ys$ ranges over sets of natural numbers,
    $\tr {\ys \in \xt_i} = \X (\ai i)_{\pix \ys}$.
  \item If $\xo$ ranges over natural numbers and $\yt$ over sets of natural
    numbers, $\tr {\xo \in \yt} = \F(1_{\pix \xo} \land 1_{\pix \yt})$.
  \item $\tr {\xo < \yo} = \F(1_{\pix \xo} \land \X \F 1_{\pix \yo})$.
  \item $\tr {\xo + \yo = \zo} = \exists \pi.\ (\X \add_\pi) \land \F(\argl_\pi \land 1_{\pix \xo}) \land
    \F(\argr_{\pi} \land 1_{\pix \yo}) \land \F(\res_\pi \land 1_{\pix \zo})$.
    \item $\tr {\xo\cdot\yo = \zo} =
    \exists \pi.\ (\X \mult_\pi) \land \F(\argl_\pi \land 1_{\pix \xo}) \land
    \F(\argr_{\pi} \land 1_{\pix \yo}) \land \F(\res_\pi \land 1_{\pix \zo})$.
  \end{itemize}

  If $\psi$ is true under some interpretation $\inter$ of $\xt_1, \ldots, \xt_n$
  as sets of sets of natural numbers,
  then the transition system $\tsysI$ defined above is a model of $\varphi'$.
  Conversely, if $\tsys \models \varphi'$ for some transition system $\tsys$,
  then for all sets $A \subseteq \Nat$ there is a path $\rho_A$ matching $A$ in $\tsys$,
  and all such paths have the same $\{\ai 1,\ldots,\ai n\}$-labelling, so we
  can define an interpretation $\inter$ of $\xt_1,\ldots,\xt_n$ by taking
  $A \in \inter(\xt_i)$ if and only if $\ai i \in \lambda(\rho_A(0))$.
  Under this interpretation $\psi$ holds, and thus $\varphi$ is true, as first- and second-order quantification in   $(\Nat, +, \cdot, <,\in)$ is mimicked by path quantification in $\tsys$.
\end{proof}

Now, we have all the tools at hand to prove the lower bound on the \hyctlstar satisfiability problem.

\begin{lem}
  \hyctlstar satisfiability is $\Sigma^2_1$-hard.
\end{lem}
\begin{proof}
  Let $N$ be a $\Sigma^2_1$ set, i.e.\
  $N = \set{x \in \nats \mid \exists x_0.\ \cdots \exists x_k.\ \psi(x, x_0, \ldots, x_k)}$ for some second-order arithmetic formula $\psi$ with existentially
  quantified third-order variables $x_i$.
  For every $n \in \Nat$, we define the sentence
  \[\varphi_n = \exists x_0.\ \cdots \exists x_k.\ \psi(n,x_0,\ldots,x_k)\, .\]
	Recall that every fixed natural number~$n$ is definable in first-order arithmetic, which is the reason we can use $n$ in $\psi$.
  
  Then $\varphi_n$ is true  if and only if  $n \in N$.
  Combining this with \autoref{lem:hyctlstar-reduction}, we obtain a
  computable function that maps any $n \in \Nat$ to a \hyctlstar formula
  $\phi'_n$ such that $n \in N$ if and only if $\phi'_n$ is satisfiable.
\end{proof}

\subsection{Variations of \texorpdfstring{\hyctlstar}{HyperCTL*} Satisfiability}

The general \hyctlstar satisfiability problem, as studied above, asks for the existence of a model of arbitrary size. 
In the $\Sigma^2_1$-hardness proof we relied on uncountable models with infinite branching. 
Hence, it is natural to ask whether satisfiability is easier when we consider restricted classes of transition systems.
In the remainder of this section, we study the following variations of satisfiability.

\begin{itemize}
\item The \hyctlstar \emph{finitely-branching satisfiability problem}: given a \hyctlstar sentence, determine whether it has a finitely-branching model.\footnote{A transition system is finitely-branching, if every vertex has only finitely many successors.}
\item The \hyctlstar \emph{countable satisfiability problem}: given a \hyctlstar sentence, determine whether it has a countable model.
\end{itemize}

But let us again begin with finite-state satisfiability, i.e., the question whether a given \hyctlstar sentence is satisfied by a finite transition system.
As for \hyltl (see \autoref{prop_hyltlfs}), finite-state satisfiability is much simpler, but still undecidable.
In fact, the lower bound is directly inherited from \hyltl while the argument for the upper bound is the same, as \hyctlstar model-checking is also decidable.

\begin{propC}[\cite{FinkbeinerRS15,FinkbeinerH16}]
\hyctlstar finite-state satisfiability is $\Sigma_1^0$-complete.
\end{propC}

Next, we show that the complexity of \hyctlstar finitely-branching satisfiability
and countable satisfiability lies between that of finite-state satisfiability
and general satisfiability:
both are equivalent to \emph{truth in second-order arithmetic}, that is,
the problem of deciding whether a given sentence of second-order arithmetic is satisfied
in the standard model~$(\nats, 0,1,+,\cdot, <,\in)$ of second-order arithmetic.

\begin{thm}\label{thm:hyctlvars}
All of the following problems are effectively interreducible:
  \begin{enumerate}
  \item\label{thm:hyctlvars:count} \hyctlstar countable satisfiability.
  \item\label{thm:hyctlvars:finbranch} \hyctlstar finitely-branching satisfiability.
  \item\label{thm:hyctlvars:sotruth} Truth in second-order arithmetic.
  \end{enumerate}
\end{thm}

To prove \autoref{thm:hyctlvars}, we show the implication (\ref{thm:hyctlvars:count}) $\Rightarrow$ (\ref{thm:hyctlvars:sotruth}) in \autoref{lem:hyctlvars:count} and the  implication (\ref{thm:hyctlvars:finbranch}) $\Rightarrow$ (\ref{thm:hyctlvars:sotruth}) in \autoref{lem:hyctlvars:finbranch}.
Then, in \autoref{lem:hyctlvars:rev} we show both converse implications simultaneously. 

We start by showing that countable satisfiability can be effectively reduced to truth in second-order arithmetic.
As every countable set is in bijection with the natural numbers, countable satisfiability asks for the existence of a model whose set of vertices is the set of natural numbers. 
This can easily be expressed in second-order arithmetic, leading to a fairly
straightforward reduction to truth in second-order arithmetic.

\begin{lem}\label{lem:hyctlvars:count}
There is an effective reduction from \hyctlstar countable
  satisfiability to truth in second-order arithmetic.
\end{lem}

\begin{proof}
  Let $\phi$ be a \hyctlstar sentence. We construct a sentence~$\phicount$ of second-order
  arithmetic such that
  $(\nats,0,1,+,\cdot,<,\in) \models \phicount$ if and only if $\phi$
  has a countable model, or, equivalently, if and only if $\phi$ has a model
  of the form $\tsys = (\nats,E,0,\lambda)$ with vertex set~$\nats$, which
 implies that
  the set~$E$ of edges is a subset of $\nats \times \nats$.
  Note that we assume (w.l.o.g.) that the initial
  vertex is~$0$. The labeling function $\lambda$ maps each natural number
  (that is, each vertex) to a set of atomic propositions.
  We assume a fixed encoding of valuations in $2^\AP$ as natural numbers in
  $\{0,\ldots,|2^{\AP}|-1\}$, so that we can equivalently view $\lambda$ as a
  function $\lambda : \nats \rightarrow \nats$ such that $\lambda(n) < \size{2^{\AP}}$
  for all $n \in \nats$. Note that binary relations over $\nats$ can be encoded by functions from natural numbers to natural numbers, and the encoding can be implemented in first-order arithmetic.

  The formula $\phicount$ is defined as
  \[
    \phicount = \exists E.\, \exists \lambda.\,
    (\forall x.\ \exists y.\ (x,y) \in E) \land
    (\forall x.\, \lambda(x) < |2^\AP|) \land 
    \phi'(E,\lambda,0) \, ,
  \]
  where $E$ is a second-order variable ranging over subsets of
  $\nats \times \nats$,
  $\lambda$ a second-order variable ranging over functions from $\nats \to  \nats$,
  and $\phi'(E,\lambda,i)$, defined below, expresses the fact that
  the transition system $(\nats,E,0,\lambda)$ is a model of $\phi$. 
  
  We use the following abbreviations:
  \begin{itemize}
  \item Given a second-order variable $f$ ranging over functions from
    $\nats$ to $\nats$, the formula
    $\mathit{path}(f,E) = \forall n.\ (f(n),f(n+1)) \in E$
    expresses the fact that $f(0) f(1) f(2) \ldots$ is
    a path in $(\nats,E,0,\lambda)$.

  \item Given second-order variables $f$ and $f'$ ranging over functions from
    $\nats$ to $\nats$ and a first-order variable $i$ ranging over natural
    numbers, we let
    \[
      \mathit{branch}(f,f',i,E) = \mathit{path}(f,E) \land \mathit{path}(f',E) \land \forall j \le i.\ 
      f(j) = f'(j) \, .
    \]
    This formula is satisfied by paths $f$ and $f'$ if $f$ and $f'$ coincide up to (and including) position $i$. We will use to restrict path quantification to those that start at a given position of a given path.
  \end{itemize}

  We define $\phi'$ inductively from $\phi$, therefore considering in general
  \hyctlstar formulas with free variables~$\pi_1,\ldots,\pi_k$, in which case the formula $\phi'$ has free variables
  $E,\lambda,f_{\pi_1},\ldots,f_{\pi_k},i$.
  The variable $i$ is interpreted as the current time point.
  If $\phi$ is a sentence, $i$ is not free in $\phi'$, as we use $0$ in that case.
Also, the translation depends on an ordering of the free variables of $\phi$, i.e.\  quantified paths start at position~$i$ of the largest variable, as path quantification depends on the context of a formula with free variables.
In the following, we indicate the ordering by the naming of the variables, i.e.\ we have  $\pi_1 <\cdots <\pi_k$.
  \begin{itemize}
   
    \item $a'_{\pi_j}(E,\lambda,f_{\pi_1},\ldots, f_{\pi_k} ,i) =
    \bigvee_{\{v \in 2^\AP \mid a \in v\}} \lambda(f_{\pi_j}(i)) = [v]$,
    where $[v]$ is the encoding of $v$ as a natural number.
     
  \item If $\phi(\pi_1, \ldots, \pi_k) = \lnot \psi$ then $\phi'(E,\lambda,f_{\pi_1},\ldots, f_{\pi_k} ,i) = \neg (\psi'(E,\lambda,f_{\pi_1},\ldots, f_{\pi_k} ,i))$.
     
  \item If $\phi(\pi_1, \ldots, \pi_k) = \psi_1 \lor \psi_2 $ then \[\phi'(E,\lambda,f_{\pi_1},\ldots, f_{\pi_k} ,i) = (\psi_1'(E,\lambda,f_{\pi_1},\ldots, f_{\pi_k} ,i))\lor(\psi_2'(E,\lambda,f_{\pi_1},\ldots, f_{\pi_k} ,i)).\]
   
  \item If $\phi(\pi_1,\ldots,\pi_k) = \X \psi$,  then we define
    \[
      \phi'(E,\lambda,f_{\pi_1},\ldots,f_{\pi_k},i) =
      \psi'(E,\lambda,f_{\pi_1},\ldots,f_{\pi_k},i+1) \, .
    \]
     
  \item If $\phi(\pi_1,\ldots,\pi_k) = \psi_1 \U \psi_2$,  then we define
    \begin{align*}
      \phi'(E,\lambda,_{\pi_1},\ldots,f_{\pi_k},i) =
      \exists j.\,
      & j \ge i \land 
       \psi_2'(E,\lambda,f_{\pi_1},\ldots,f_{\pi_k},j) \land {} \\
      & \forall j'.\  (i \le j' < j \rightarrow
        \psi_1'(E,\lambda,f_{\pi_1},\ldots,f_{\pi_k},j')) \, .
    \end{align*}
   
  \item
    If $\phi = \exists \pi_1. \psi(\pi_1)$ is a sentence, then we define
    \[
      \phi'(E,\lambda) = \exists f_{\pi_1}.\ 
      \mathit{path}(f_{\pi_1},E) \land f_{\pi_1}(0) = 0 \land
      \psi'(E,\lambda,f_{\pi_1},0) \, .
    \]
    Recall that~$f_{\pi_1}$ ranges over functions from $\nats$ to $\nats$ and note that the formula requires $f$ to encode a path and to start at the initial vertex~$0$.
    
    If $\phi(\pi_1,\ldots,\pi_k) = \exists \pi_{k+1}.\ \psi(\pi_1,\ldots,\pi_k,\pi_{k+1})$
    with $k > 0$, then we define
    \[
      \phi'(E,\lambda,f_{\pi_1},\ldots,f_{\pi_k},i) = 
        \exists f_{\pi_{k+1}}.\  \mathit{branch}(f_{\pi_{k+1}},f_{\pi_k},i,E) \land
        \psi'(E,\lambda,f_{\pi_1},\ldots,f_{\pi_k},f_{\pi_{k+1}},i) \, .
      \]
      Here, we make use of the ordering of the free variables of $\phi$, as the translated formula requires the function assigned to $f_{\pi_{k+1}}$ to encode a path branching of the path encoded by the function assigned to $f_{\pi_{k}}$.

    \item
      If $\phi = \forall \pi_1.\ \psi(\pi_1)$ is a sentence, then we define
      \[
        \phi'(E,\lambda) = \forall f_{\pi_1}.\
        \left(\mathit{path}(f_{\pi_1},E) \land f_{\pi_1}(0) = 0 \right) \rightarrow
        \psi'(E,\lambda,f_{\pi_1},0) \, .
      \]
      If $\phi(\pi_1,\ldots,\pi_k) = \forall \pi_{k+1}.\ \psi(\pi_1,\ldots,\pi_k,\pi_{k+1})$
      with $k > 0$,  then we define
    \[
      \phi'(E,\lambda,f_{\pi_1},\ldots,f_{\pi_k},i) = 
        \forall f_{\pi_{k+1}}.\  \mathit{branch}(f_{\pi_{k+1}},f_{\pi_k},i,E) \rightarrow
        \psi'(E,\lambda,f_{\pi_1},\ldots,f_{\pi_k},f_{\pi_{k+1}},i)
      \]
  \end{itemize}
  
Now, $\phi$ has a countable model if and only if the second-order sentence~$\phicount$ is true in $(\nats,0,1,+,\cdot,<,\in)$.
\end{proof}

Since every finitely-branching model has countably many vertices that are reachable from the initial vertex, the previous proof can
be easily adapted for the case of finitely-branching satisfiability.

\begin{lem}\label{lem:hyctlvars:finbranch}
There is an effective reduction from \hyctlstar finitely-branching
  satisfiability to truth in second-order arithmetic.
\end{lem}

\begin{proof}
  Let $\phi$ be a \hyctlstar sentence. We construct a second-order
  arithmetic formula~$\phifb$ such that
  $(\nats,0,1,+,\cdot,<,\in) \models \phifb$ if and only if $\phi$ has a
  finitely-branching model, which we can again assume without loss of
  generality to be of the form $\tsys = (\nats,E,0,\lambda)$, where the set of
  vertices is $\nats$, the set~$E$ of edges is a subset of
  $\nats \times \nats$, the initial vertex is $0$, and the labeling function
  $\lambda$ is encoded as a function from  $\nats$ to $ \nats$.

  The formula $\phifb$ is almost identical to $\phicount$, only adding the
  finite branching requirement:
  \begin{align*}
    \phifb = \exists E.\ \exists \lambda.\ &
    (\forall x.\ \exists y.\ (x,y) \in E) \land
    (\forall x.\ \exists y.\ \forall z.\ (x,z) \in E \rightarrow z < y) \land\\
    &     (\forall x.\ \lambda(x) < |2^\AP|) \land \\
   & \phi'(E,\lambda,0) \, . \qedhere     
  \end{align*}
  
\end{proof}

Now, we consider the converse, i.e.\ that truth in second-order arithmetic can be reduced to countable and finitely-branching satisfiability. To this end, we adapt the $\Sigma_1^2$-hardness proof for \hyctlstar. 
Recall that we constructed a formula whose models contain all $\set{0,1}$-labelled paths, which we used to encode the subsets of $\nats$. 
In that proof, we needed to ensure that the initial vertices of all these paths are pairwise different in order to encode existential third-order quantification, which resulted in uncountably many successors of the initial vertex.
Also, we used the traces in $\Top$ to encode arithmetic operations. 

Here, we only have to encode first- and second-order quantification, so we can drop the requirement on the initial vertices of the paths encoding subsets, which simplifies our construction and removes one source of infinite branching.
However, there is a second source of infinite branching, i.e.\ the infinitely many traces in $\Top$ which all start at successors of the initial vertex. This is unavoidable: To obtain formulas that always have finitely-branching models, we can no longer work with $\Top$.
We begin by explaining the reason for this and then explain how to adapt the construction to obtain the desired result. 

Recall that we defined $\Top$ over $\AP = \{\argl,\argr,\res,\add,\mult\}$ as the set of all traces $t \in {(2^{\AP})}^\omega$ such that 
  \begin{itemize}
  \item there are unique $n_1,n_2,n_3 \in \nats$ with $\argl \in t(n_1)$,
    $\argr \in t(n_2)$, and $\res \in t(n_3)$, and
  \item either $\add \in t(n)$ and $\mult \notin t(n)$  for all $n$ and $n_1+n_2 = n_3$,
    or $\mult \in t(n)$ and $\add \notin t(n)$  for all $n$ and $n_1 \cdot n_2 = n_3$.
  \end{itemize}

An application of Kőnig's Lemma~\cite{konig} shows that there is no finitely-branching transition system whose set of traces is $\Top$. 
The reason is that $\Top$ is not (topologically) closed (see definitions below), while the set of traces of a finitely-branching transition system is always closed.

Let $\prefs{t} \subseteq (\pow{\ap})^*$ denote the set of finite prefixes of a trace~$t \in (\pow{\ap})^\omega$.
Furthermore, let $\prefs{T} = \bigcup_{t \in T}\prefs{t}$ be the set of finite prefixes of a set~$T \subseteq (\pow{\ap})^\omega$ of traces.
The closure~$\closure{T} \subseteq (\pow{\ap})^\omega$ of such a set~$T$ is defined as 
\[
\closure{T} = \set{t \in (\pow{\ap})^\omega \mid \prefs{t} \subseteq \prefs{T}}.
\]
For example, $\set{\add}^\omega \in \closure{\Top}$ and $\set{\mult}^*\set{\argr,\mult}\set{\mult}^\omega \subseteq \closure{\Top}$.
Note that we have $T \subseteq \closure{T}$ for every $T$. 
As usual, we say that $T$ is closed if $T = \closure{T}$.

Let $\ap$ be finite and let $T \subseteq (\pow{\ap})^\omega$ be closed.
Furthermore, let $\tsys(T)$ be the finitely-branching transition system~$(\prefs{T}, E, \epsilon, \lambda)$ with \[E = \set{(w,wv) \mid wv \in \prefs{T} \text{ and } v \in \pow{\ap} },\]
$\lambda(\epsilon) = \emptyset$, and $\lambda(wv) = v$ for all $wv \in \prefs{T}$ with $v \in \pow{\ap}$.

\begin{rem}
The set of traces of paths of $\tsys(T)$ starting at the successors of the initial vertex~$\epsilon$ is exactly $T$.
\end{rem}

In the following, we show that we can replace the use of $\Top$ by $\closure{\Top}$ and still capture addition and multiplication in \hyltl.
We begin by characterising the difference between $\Top$ and $\closure{\Top}$ and then show that $\closure{\Top}$ is also the unique model of some \hyltl sentence~$\phiopc$.

Intuitively, a trace is in $\closure{\Top} \setminus \Top$ if at least one of the arguments (the propositions~$\argl$ and $\argr$) are missing. 
In all but one case, this also implies that $\res$ does not occur in the trace, as the position of $\res$ is (almost) always greater than the positions of the arguments. 
The only exception is when $\mult$ holds and $\res$ holds at the first position, i.e.\ in the limit of traces encoding $0\cdot n = n$ for $n$ tending towards infinity.

Let $\diff$ be the set of traces~$t$ over $\AP = \{\argl,\argr,\res,\add,\mult\}$ such that
\begin{itemize}
    \item for each $a \in \set{\argl,\argr,\res}$ there is at most one	~$n$ such that $a \in t(n)$, and
	\item either $\add \in t(n)$ and $\mult \notin t(n)$  for all $n$,
    or $\mult \in t(n)$ and $\add \notin t(n)$  for all $n$, 
    \item there is at least one $a \in \set{\argl,\argr}$  such that $a \notin t(n)$ for all $n$.
    \item Furthermore, if there is an $n$ such that $\res \in t(n)$, then $\mult \in t(0)$, $n = 0$, and either $\argl\in t(0)$ or $\argr\in t(0)$. 
\end{itemize}

\begin{rem}
$\closure{\Top} \setminus \Top = D$.	
\end{rem}

Now, we show the analogue of \autoref{lem:phiop} for $\closure{\Top}$.

\begin{lem}\label{lem:phiopc}
  There is a \hyltl sentence $\phiopc$ which has $\closure{\Top}$ as unique model.
\end{lem}

\begin{proof}
We adapt the formula~$\phiop$ presented in the proof of \autoref{lem:phiop} having $\Top$ as unique model.
Consider the conjunction of the following \hyltl sentences:
\begin{enumerate}

	\item For every trace~$t$ and every $a \in \set{\argl,\argr,\res}$ there is at most one $n$ such that $a \in t(n)$:
\[
		\forall \pi.\ \bigwedge_{a \in \set{\argl,\argr,\res}}(\G\neg a_\pi) \vee (\neg a_\pi)\U (a_\pi \wedge \X\G\neg a_\pi)
\]
\item For all traces $t$: If both $\argl$ and $\argr$ appear in $t$, then also $\res$ (this captures the fact that the position of $\res$ is determined by the positions of $\argl$ and $\argr$):

\[
\forall \pi.\ (\F\argl_\pi \wedge \F\argr_\pi) \rightarrow \F\res_\pi\]

	\item Every trace~$t$ satisfies either $\add \in t(n)$ and $\mult\notin t(n)$ for all $n$ or $\mult \in t(n)$ and $\add \notin t(n)$ for all $n$:
	\[
	\forall \pi.\ \G(\add_\pi\wedge\neg\mult_\pi) \vee \G(\mult_\pi\wedge\neg\add_\pi)
	\]	
	
	\item For all traces~$t$: If there is an $a \in \set{\argl,\argr}$ such that $a \notin t(n)$ for all $n$, but $\res \in t(n_3)$ for some $n_3$, then $\set{\mult, \res} \subseteq t(0)$ and $\set{\argl,\argr} \cap t(0) \neq \emptyset$:
	\[
	\forall \pi.\ \left(\F\res_\pi \wedge \bigvee_{a \in \set{\argl,\argr}}\G\neg a_\pi\right)\rightarrow \left(\mult_\pi \wedge \res_\pi \wedge \bigvee_{a \in \set{\argl,\argr}}a\right)
	\]

\end{enumerate}
 We again only consider traces satisfying these formulas in the remainder of the proof, as all others are not part of a model. Also, we again speak of addition traces (if $\add$ holds) and multiplication traces (if $\mult $ holds).
 
	Furthermore, if a trace satisfies the (guard) formula~$\phi_g = \F \arg_1 \wedge \F\argr$, then it encodes two unique arguments (given by the unique positions $n_1$ and $n_2$ such that $\argl\in t(n_1)$ and $\argr\in t(n_2)$. As the above formulas are satisfied, such a trace also encodes a result via the unique position~$n_3$ such that $\res \in t(n_3)$. 
	
As before, we next express that every combination of inputs is present:
\begin{enumerate}
\setcounter{enumi}{4}
	
	\item\label{itemcomplstartcl}
     There are two traces with both arguments being zero, one for addition and one for multiplication:
\[\bigwedge_{a \in \set{\add,\mult}}\exists \pi.\ a_\pi \wedge \argl_\pi \wedge \argr_\pi\]

\item\label{itemcomplstepcl} For every trace encoding the arguments~$n_1$ and $n_2$, the argument combinations~$(n_1+1, n_2)$ and $(n_1, n_2+1)$ are also represented in the model, again both for addition and multiplication (here we rely on the fact that either $\add$ or $\mult$ holds at every position, as specified above). Note however, that not every trace will encode two inputs, which is why we have to use the guard~$\phi_g$.
\begin{align*}
\forall \pi.\ \phi_g \rightarrow  \exists\pi_1, \pi_2.\ \left( \bigwedge_{i \in\set{1,2}} \add_\pi \leftrightarrow \add_{\pi_i} \right)\wedge 
&\F(\argl_\pi \wedge \X\argl_{\pi_1}) \wedge \F(\argr_\pi \wedge \argr_{\pi_1}) \wedge \\
&
\F(\argl_\pi \wedge \argl_{\pi_2}) \wedge \F(\argr_\pi \wedge \X\argr_{\pi_2})	
\end{align*}  
\end{enumerate}	
Every model of these formulas contains a trace representing each possible combination of arguments, both for multiplication and addition. 

To conclude, we need to express that the result in each trace is correct by again capturing the inductive definition of addition in terms of repeated increments and the inductive definition of multiplication in terms of repeated addition.
The formulas differ from those in the proof of \autoref{lem:phiop} only in the use of the guard~$\phi_g$.

	\begin{enumerate}
	  \setcounter{enumi}{6}
	\item\label{itemcorrfirstcl} For every trace~$t$: if $\set{\add,\argl} \subseteq t(0)$ and $\argr$ appears in $t$ then $\argr$ and $\res$ have to hold at the same position (this captures $0 + n = n$):
	\[
	\forall \pi.\ (\phi_g\wedge \add\wedge \argl_\pi ) \rightarrow \F(\argr_\pi \wedge \res_\pi)
	\]
	\item For each trace~$t$ with $\add\in t(0)$, $\argl \in t(n_1)$,
    $\argr \in t(n_2)$, and $\res \in t(n_3)$ such that $n_1 > 0$ there is a trace~$t'$ such that $\add\in t'(0)$, $\argl \in t'(n_1-1)$,
    $\argr \in t'(n_2)$, and $\res \in t'(n_3-1)$ (this captures $n_1 + n_2 = n_3 \Leftrightarrow n_1-1 + n_2 = n_3-1$ for $n_1 > 0$):
    \\[.5em]$\displaystyle
        \forall \pi.\ \exists\pi'.\ (\phi_g\wedge\add_\pi \wedge \neg \argl_\pi) \rightarrow \\\phantom{x}\hfill
        \left(\add_{\pi'} \wedge  \F( \X\argl_{\pi} \wedge \argl_{\pi'} ) \wedge \F(\argr_\pi \wedge \argr_{\pi'}) \wedge \F(\X\res_{\pi} \wedge \res_{\pi'}) \right)$\\[-,5em]
    
        \item For every trace~$t$: if $\set{\mult,\argl} \subseteq t(0)$ then also $\res \in t(0)$ (this captures $0 \cdot n = 0$):
	\[
	\forall \pi.\ (\mult\wedge \argl_\pi ) \rightarrow  \res_\pi
	\]
	\item\label{itemcorrlastcl} Similarly, for each trace~$t$ with $\mult \in t(0)$, $\argl \in t(n_1)$,
    $\argr \in t(n_2)$, and $\res \in t(n_3)$ such that $n_1 > 0$ there is a trace~$t'$ such that $\mult \in t'(0)$, $\argl \in t'(n_1-1)$,
    $\argr \in t'(n_2)$, and $\res \in t'(n_3-n_2)$.
   The latter requirement is expressed by the existence of a trace~$t''$ with $\add\in t''(0)$, $\argr \in t''(n_2)$, $\res \in t''(n_3)$, and $\argl$ holding in $t''$ at the same time as $\res$ in $t'$, which implies $\res \in t'(n_3-n_2)$. Altogether, this captures $n_1 \cdot n_2 = n_3 \Leftrightarrow (n_1-1) \cdot n_2 = n_3-n_2$ for $n_1 > 0$.
    \begin{align*}
    \forall \pi.\ \exists\pi',\pi''.\ &(\phi_g\wedge\mult_\pi \wedge  \neg \argl_\pi) \rightarrow   \left(\mult_{\pi'} \wedge \add_{\pi''} \wedge\right.  \\ 
&    
    \F(\X\argl_{\pi} \wedge \argl_{\pi'} ) \wedge 
    \F(\argr_\pi \wedge \argr_{\pi'} \wedge \argr_{\pi''}) \wedge \\ 
&  \left.    \F(\res_{\pi'} \wedge \argl_{\pi''}) \wedge 
  \F(\res_{\pi} \wedge \res_{\pi''}
    \right)
    \end{align*}
\end{enumerate}

Now, $\closure{\Top}$ is a model of the conjunction~$\phiopc$ of these ten formulas. Conversely, every model of $\phiopc$ contains all possible combinations of arguments (both for addition and multiplication) due to Formulas~(\ref{itemcomplstartcl}) and (\ref{itemcomplstepcl}).
Now, Formulas~(\ref{itemcorrfirstcl}) to (\ref{itemcorrlastcl}) ensure that the result is \emph{correct} on these traces.
Furthermore, all traces in $\diff$, but not more, are also contained due to the first four formulas. 
 Altogether, this implies that $\closure{\Top}$ is the unique model of $\phiopc$.
\end{proof}

We are now ready to prove the lower bounds for \hyctlstar countable and
finitely-branching satisfiability.

\begin{lem}\label{lem:hyctlvars:rev}
  There is an effective reduction from truth in second-order arithmetic to
  \hyctlstar countable and finitely-branching satisfiability.
\end{lem}

\begin{proof}
We proceed as in the proof of \autoref{lem:hyctlstar-reduction}.
  Given a sentence~$\phi$ of second-order arithmetic we construct a \hyctlstar formula $\varphi'$ such that
  $(\Nat, +, \cdot, <,\in)$ is a model of $\varphi$ if and only if
  $\varphi'$ is satisfied by a countable and finitely-branching model.

As before, we represent sets of natural numbers
  as infinite paths with labels in $\{0,1\}$, so quantification over sets of natural numbers and natural numbers is captured by path quantification.
The major difference between our proof here and the one of \autoref{lem:hyctlstar-reduction} is that we do not need to deal with third-order quantification here.
This means we only need to have every possible $\{0,1\}$-labeled path in our models, but not with pairwise distinct initial vertices.
In particular, the finite (and therefore finitely-branching) transition system~$\tsys_f$ depicted in \autoref{fig:phisetc} has all such paths.

\begin{figure}
    \centering
\def\dist{2.0cm}
\begin{tikzpicture}[->,
>=stealth', 
level/.style={sibling distance = 3cm/#1, level distance = \dist}, 
scale=0.7,
transform shape,
grow=right,thick]

\node (init) [state] {\phantom{0}}
child {
    node (t0) [state] {1} 
}
child {
    node (t1) [state] {0} 
}
;

\draw[->] (-1,0) -- (init);
\path[->] 
(t0) edge[bend left] (t1)
(t1) edge[bend left] (t0)
(t0) edge[loop right] ()
(t1) edge[loop right] ()
;

\end{tikzpicture}
    \caption{A depiction of $\tsys_f$. All vertices but the initial one are labelled by $\fbt$.
}%

    \label{fig:phisetc}
\end{figure}

For arithmetical operations, we rely on the \hyltl sentence~$\phiopc$ from
  \autoref{lem:phiopc}, with its unique model~$\closure{\Top}$, and the transition system~$\tsys(\closure{\Top})$, which is countable, finitely-branching, and whose set of traces starting at the successors of the initial vertex is exactly $\closure{\Top}$.
  We combine $\tsys_f$ and $\tsys(\closure{\Top})$ by identifying their respective initial vertices, but taking the disjoint union of all other vertices.
  The resulting transition system~$\tsysc$ contains all traces encoding the subsets of the natural numbers as well as the traces required to model arithmetical operations.
  Furthermore, it is still countable and finitely-branching.
  
  \smallskip
  
  Let $\AP = \{0,1,\fbt,\argl,\argr,\res,\mult,\add\}$.
  Using parts of the formula~$\phiset$ defined in \autoref{lem:phiset}
  and the formula~$\phiopc$ defined in \autoref{lem:phiopc}, it is not difficult to construct a formula
  $\phizc$ such that:
  \begin{itemize}
  \item The transition system $\tsysc$ is a model of $\phizc$.
  \item Conversely, in any model $\tsys = (V,E,v_\initmark, \lambda)$ of $\phizc$, the following conditions
    are satisfied:
  \begin{enumerate}
    \item
      For every path $\rho$ starting at a $\fbt$-labelled successor of
      the initial vertex $v_\initmark$, every vertex $\rho(i)$ with $i \ge 0$ has a label
      $\lambda(\rho(i)) = \{\fbt,0\}$ or $\lambda(\rho(i)) = \{\fbt,1\}$.
      
    \item For every $A \subseteq \Nat$, there exists a $\fbt$-labelled path
      $\rho_A$ starting at a successor of $v_\initmark$ 
      such that $1 \in \lambda(\rho_A(i))$ if  $i \in A$,
      and $0 \in \lambda(\rho_A(i))$ if  $i \notin A$.
     \item For every path $\rho$ starting at an $\add$- or $\mult$-labelled
      successor of the initial vertex, the label sequence
      $\lambda(\rho(0)) \lambda(\rho(1)) \cdots$ of $\rho$ is in $\closure{\Top}$.

    \item Conversely, for every trace $t \in \closure{\Top}$, there exists a path
      $\rho$ starting at a successor of the initial vertex such that
      $\lambda(\rho(0)) \lambda(\rho(1)) \cdots = t$.
    \end{enumerate}
  \end{itemize}

  We then let $\varphi' = \phizc  \land \tr  \varphi$, where $\tr  \phi$ is defined inductively from $\varphi$ as in the proof of \autoref{lem:hyctlstar-reduction}:
  \begin{itemize}
  \item $\tr {\psi_1 \lor \psi_2} = \tr {\psi_1} \lor \tr {\psi_2}$.
  \item 
    $\tr {\lnot \psi_1} = \lnot \tr {\psi_1}$.
  \item If $x$ ranges over sets of natural numbers,
    \[\tr {\exists x.\ \psi_1} =
    \exists \pix x.\ ((\X \fbt_{\pix x}) \land \tr  {\psi_1}),\]
    and
    \[\tr {\forall x.\ \psi_1} =
    \forall \pix x.\ ((\X \fbt_{\pix x}) \rightarrow \tr  {\psi_1}).\]

  \item If $x$ ranges over natural numbers, 
    \[\tr {\exists x.\ \psi_1} = \exists \pix x.\ ((\X \fbt_{\pix x}) \land
    \X (0_{\pix x} \U (1_{\pix x} \land \X\G 0_{\pix x})) \land \tr  {\psi_1}),\]
    and
    \[\tr {\forall x.\ \psi_1} = \forall \pix x.\ ((\X \fbt_{\pix x}) \land
    \X (0_{\pix x} \U (1_{\pix x} \land \X\G 0_{\pix x})) \rightarrow \tr  {\psi_1}).\] 
    Here, the subformula~$0_{\pix x} \U (1_{\pix x} \land \X\G 0_{\pix x})$ expresses that there is a single $1$ on the trace assigned to $\pix x$, i.e.\ the path represents a singleton set.
  \item If $\xo$ ranges over natural numbers and $\yt$ over sets of natural
    numbers, $\tr {\xo \in \yt} = \F(1_{\pix \xo} \land 1_{\pix \yt})$.
  \item $\tr {\xo < \yo} = \F(1_{\pix \xo} \land \X \F 1_{\pix \yo})$.
  \item $\tr {\xo + \yo = \zo} = \exists \pi.\ (\X \add_\pi) \land \F(\argl_\pi \land 1_{\pix \xo}) \land
    \F(\argr_{\pi} \land 1_{\pix \yo}) \land \F(\res_\pi \land 1_{\pix \zo})$.
   \item  $\tr {\xo\cdot\yo = \zo} =
    \exists \pi.\ (\X \mult_\pi) \land \F(\argl_\pi \land 1_{\pix \xo}) \land
    \F(\argr_{\pi} \land 1_{\pix \yo}) \land \F(\res_\pi \land 1_{\pix \zo})$.
  \end{itemize}

  If $\varphi$ is true in $(\Nat, +, \cdot, <,\in)$,
  then the countable and finitely-branching transition system~$\tsysc$ defined above is a model of $\varphi'$.
  Conversely, if $\tsys \models \varphi'$ for some transition system $\tsys$,
  then for all sets $A \subseteq \Nat$ there is a path $\rho_A$ matching $A$ in $\tsys$ and trace quantification in $\tsys$ mimics first- and second-order in $(\Nat, +, \cdot, <,\in)$.
  Thus, $\varphi$ is true in $(\Nat, +, \cdot, <,\in)$.
\end{proof}

Note that the preceding proof shows that even \hyctlstar bounded-branching satisfiability is equivalent to truth in second-order arithmetic, i.e., the question of whether a given sentence is satisfied by a transition system where each vertex has at most $k$ successors, for some uniform $k \in \nats$.

%^^^^^^^^^^^^^^^^^^^^^^^^^^^^^^^^^^^^^^^^^^^^^^
%^^^^^^^^^^^^^^^^^^^^^^^^^^^^^^^^^^^^^^^^^^^^^^
\section{Related Work}
\label{sec:relatedwork}

The \hyltl and \hyctlstar model checking problems have been shown decidable in the first paper introducing these logics~\cite{ClarksonFKMRS14} and their exact complexity (and that of variants) has been determined in a series of further works~\cite{FinkbeinerRS15,Rabe16,MZ20}. 
The \hyltl satisfiability problem has been shown undecidable by Finkbeiner and Hahn~\cite{FHH18}, but no upper bounds on the complexity were known.
Similarly, Rabe has shown that \hyctlstar satisfiability is $\Sigma_1^1$-hard~\cite{Rabe16}, but again no upper bounds were known.
Here, we settle the exact complexity of satisfiability for both logics as well as that of some variants. 

Further complexity results have been obtained by Bonakdarpour and Finkbeiner for monitoring~\cite{DBLP:conf/csfw/BonakdarpourF18}, by Finkbeiner et al.\ for synthesis~\cite{FinkbeinerHLST20}, and by Bonakdarpour and Sheinvald for standard automata-theoretic problems for hyperproperties represented by automata~\cite{DBLP:journals/iandc/BonakdarpourS23}, while Winter and Zimmermann showed that the existence of computable Skolem functions for \hyltl is decidable~\cite{WZ}. Such
functions yield explanations and counterexamples for \hyltl model checking.

The techniques and results presented here have been generalized to second-order \hyltl~\cite{DBLP:conf/cav/BeutnerFFM23}, i.e.\ \hyltl with quantification over \emph{sets} of traces.
This logic allows to express important hyperproperties like common knowledge in multi-agent systems and asynchronous hyperproperties, which are not expressible in \hyltl.
Second-order quantification increases the already high expressiveness considerably: Satisfiability, finite-state satisfiability and model checking are all equivalent to truth in third-order arithmetic~\cite{sohyltlcomplexity}. 
The intuitive reason is that second-order quantification over traces corresponds to quantification over sets of sets of natural numbers (as traces can encode characteristic sequences of such sets) and we have presented here an \myquot{implementation} of addition and multiplication in \hyltl. These ingredients yield the lower bounds while an embedding of \hyltl in third-order arithmetic yields the matching upper bounds. 

Similarly, the techniques and results presented here have been generalized to \hyqptl as well, which extends \hyltl by quantification over propositions~\cite{Rabe16}. 
With uniform quantification, \hyqptl satisfiability is equivalent to truth in second-order arithmetic~\cite{regaud2024complexityhyperqptl} while finite-state satisfiability and model-checking have the same complexity as for \hyltl~\cite{FinkbeinerH16,Rabe16}. Non-uniform quantification makes \hyqptl as expressive as second-order \hyltl~\cite{regaud2024complexityhyperqptl}, which implies that all three problems are equivalent to truth in third-order arithmetic.

The specification and verification of asynchronous hyperproperties, one of the motivations for studying second-order \hyltl, have also been addressed by introducing dedicated (first-order) extensions of \hyltl. In fact, there is a wide range of such logics~\cite{BartocciHNC23,DBLP:conf/cav/BaumeisterCBFS21,BeutnerF23,DBLP:journals/corr/abs-2404-16778,DBLP:conf/lics/BozzelliPS21,DBLP:conf/concur/BozzelliPS22,DBLP:journals/pacmpl/GutsfeldMO21,DBLP:conf/tacas/HsuBFS23,DBLP:conf/mfcs/KontinenSV23,KontinenSV24,DBLP:conf/mfcs/KrebsMV018}).
While some of these works contain partial complexity results, there is currently no full picture of the complexity of the standard verification problems for these logics.

Similarly, logics for probabilistic hyperproperties have been introduced~\cite{DBLP:conf/qest/AbrahamB18,DBLP:conf/nfm/DobeWABB22,DBLP:conf/atva/AbrahamBBD20,DimitrovaFT20}, again with partial complexity results.

Finally, there is another approach towards specifying hyperproperties, namely using team semantics for temporal logics~\cite{DBLP:conf/mfcs/KrebsMV018,teamltlarxiv}:
Lück~\cite{Luck20} studied the complexity of satisfiability and model checking for \teamltl  with Boolean negation and proved both problems to be equivalent to truth in third-order arithmetic. Kontinen and Sandström~\cite{DBLP:conf/wollic/KontinenS21} generalized this result and showed that any logic between \teamltl with Boolean negation and second-order logic inherits the same complexity results. However, the complexity of \teamltl model checking without Boolean negation is a longstanding open problem~\cite{DBLP:conf/mfcs/KrebsMV018,teamltlarxiv}.

%^^^^^^^^^^^^^^^^^^^^^^^^^^^^^^^^^^^^^^^^^^^^^^
%^^^^^^^^^^^^^^^^^^^^^^^^^^^^^^^^^^^^^^^^^^^^^^
\section{Conclusion}
\label{sec:conc}

In this work, we have settled the complexity of the satisfiability problems for \hyltl and \hyctlstar. In both cases, we significantly increased the lower bounds, i.e.\ from $\Sigma_1^0$ and $\Sigma_1^1$ to  $\Sigma_1^1$ and $\Sigma_1^2$, respectively, and presented the first upper bounds, which are tight in both cases.
Along the way, we also determined the complexity of restricted variants, e.g.\ \hyltl satisfiability restricted to ultimately periodic traces (or, equivalently, to finite traces) is still $\Sigma_1^1$-complete while \hyltl and \hyctlstar satisfiability restricted to finite transition systems is $\Sigma_1^0$-complete.
Furthermore, we proved that both countable and the finitely-branching satisfiability for \hyctlstar are as hard as truth in second-order arithmetic.
As a key step in our proofs, we showed a tight bound of $\contcard$ on the size of minimal models for satisfiable \hyctlstar sentences.
Finally, we also showed that deciding membership in any level of the \hyltl quantifier alternation hierarchy is $\Pi_1^1$-complete. 

\section*{Acknowledgment}
   \noindent This work was partially funded by EPSRC grants EP/S032207/1 and EP/V025848/1 and DIREC – Digital Research Centre Denmark.
  We thank Karoliina Lehtinen and Wolfgang Thomas for fruitful discussions.

%^^^^^^^^^^^^^^^^^^^^^^^^^^^^^^^^^^^^^^^^^^^^^^
%^^^^^^^^^^^^^^^^^^^^^^^^^^^^^^^^^^^^^^^^^^^^^^
%%   Bibliography
\bibliographystyle{alphaurl}
\bibliography{references}

\end{document}